%% file: main.tex
\documentclass[11pt]{article}

\input{preamble}
\title{Learning Random Quantum Circuits and the \\  Emergence of Pseudorandomness}

\author{
Srinivasan Arunachalam\\[2mm]
IBM Research\\
\small Silicon Valley Lab\\
\small \texttt{Srinivasan.Arunachalam@ibm.com}
\and 
Qizhao Huang \qquad \qquad   Makrand Sinha\\[2mm]
University of Illinois\\\small Urbana-Champaign\\
\small \texttt{qizhao2@illinois.edu}, \texttt{msinha@illinois.edu}
}
\date{}

\begin{document}
\maketitle 
\begin{abstract}
\noindent We give an efficient algorithm for learning $k$-dimensional brickwork random
quantum circuits using only copies of the output state obtained by applying $U$ to the all-zero input. For a depth-$d$ circuit on $n$ sites with random $2\ell$-qubit gates, the algorithm learns the original circuit $U$ with high probability in \(\operatorname{poly}(n,2^{\ell d})\) time for every constant dimensional lattice. In particular, the algorithm is polynomial time as long as $\ell d = O(\log n)$. In one dimension, this reaches the natural boundary suggested by
pseudorandomness: pseudorandom states require $\ell d=\omega(\log n)$, and structured cryptographic constructions suggest that this scale may be achievable from above. In higher dimensions, it remains plausible that the
same $\ell d=\omega(\log n)$ scale remains the threshold for ancilla-free pseudorandomness, and our results help clarify the conditions under which pseudorandomness can arise in this setting.

\vspace*{8pt}
\noindent The main idea behind our algorithm is a local correlation criterion that identifies gates in the final layer without learning their entire backward light cones, avoiding a bottleneck in the previous approaches. A key technical ingredient is a Carbery-Wright type anticoncentration inequality for low-degree polynomials of Haar random unitaries whose small-ball exponent is independent of the matrix dimension.  The dimension-independent exponent is crucial for handling gates of growing locality. This anticoncentration result may also be of independent interest.

\end{abstract}
% \newpage 
% \setcounter{tocdepth}{3}
% \begingroup
% \renewcommand{\baselinestretch}{0.2} % Multiplier to squeeze lines closer together
% {\footnotesize  \tableofcontents}
% \endgroup

\newpage

\input{introduction}

\paragraph{Acknowledgements.} We thank Bill Fefferman, Soumik Ghosh, Noam Lifshitz,  Dan Mikulincer and Joe Slote for helpful conversations. MS acknowledges support from the NSF award QCIS-FF: Quantum Computing \& Information Science Faculty Fellow at the University of Illinois Urbana-Champaign (NSF 1955032) and from an IBM-Illinois Discovery Accelerator Institute (IIDAI) research grant. QH is supported by an IBM-Illinois Discovery Accelerator Institute (IIDAI) research grant. Parts of this work were done while some of the authors were visiting the Simons Institute for the Theory of Computing.

\section{LLM Methodology} 

We follow the guidelines developed by the AI+TCS working group at the Simons Institute for the Theory of Computing, UC Berkeley~\cite{TCSAI} for transparently reporting the usage of LLM-based tools. Departing from one of the recommendations in the working group report, throughout the discussion below, we refer generically to LLM-based tools rather than naming particular systems to avoid implying an endorsement.

\begin{center}
\begin{tabular}{@{}l@{\hspace{0.5em}}c@{\hspace{5em}}l@{\hspace{0.5em}}c@{}}
\multicolumn{4}{c}{\textbf{Activities LLM-based tools were used for}} \\[12pt]
Asking the research question & \no &
Coming up with the approach & \nostar \\[2pt]
Proof development & \nostar &
Writing and exposition & \no \\[2pt]
Checking for bugs & \no &
Other supporting tasks & \yes
\end{tabular}
\end{center}

We elaborate on these below.

\textbf{Asking the research question.} The questions and many of the ideas in this paper have been development since early 2025 and LLM-based tools were not involved in this process.

\textbf{Coming up with the approach/Proof development.} We have opted to answer ``No*'' for both of these, as explained below. Many of the key ideas in this paper were developed by the authors, while for one key step the main benefit of the LLM-based tools was to point us toward relevant mathematical ideas and techniques in the literature that we were not previously aware of. The authors then developed the resulting proofs by consulting the literature and adapting these techniques to our setting, in much the same way one might adapt an argument from another paper. We describe the developments in detail next.

The main idea for proving the improved anticoncentration bound was provided by an LLM. SA first found a reduction of anticoncentration on $\U(m)$ to $\U(1)^m$ and queried an LLM in June 2026 about whether anticoncentration for low-degree Laurent polynomials was known in the mathematical literature. It suggested looking at univariate extremal inequalities and pointed us to the work of~\cite{NSV}. We then asked Joe Slote and Noam Lifshitz on how to prove the Remez inequality for our family of functions using the work of~\cite{NSV}. Both provided a short LLM-generated proof that which gave the high level sketch of \cref{thm:remez}, which we verified, simplified and rewrote entirely here. Initially, this approach led to a small-ball exponent with a polylogarithmic dependence on the dimension. Subsequently, in August 2026 MS prompted a newer LLM version to look for an example showing that a logarithmic dependence was necessary. Instead, it pointed to the Weyl Integration formula, after which the authors developed the interpolation-based proofs presented here that obtained a dimension-free small-ball exponent. Throughout this process, LLM based tools were useful for searching and understanding the literature, identifying relevant proof techniques and rapidly testing the effects of different assumptions in intermediate lemmas. All other ideas including the learning algorithm were developed by the authors. 

\textbf{Writing and exposition.} The paper was written entirely by the authors in the old-fashioned way. LLM based tools were used occasionally for suggestions, generating tables or figures, LaTeX assistance and proofreading. The authors decided whether and how to incorporate suggestions and no changes were applied automatically. All mathematical arguments and citations were independently verified by the authors, who take full responsibility for the paper.

\textbf{Checking for bugs.} LLM based tools did not contribute to identifying or correcting consequential errors in the paper.

\textbf{Other supporting tasks.} As discussed before LLM based tools were used occasionally for suggestions, generating tables or figures, LaTeX assistance, bibliography searches and proofreading. The authors decided whether and how to incorporate suggestions and no changes were applied automatically.

\input{preliminaries}

\input{anticoncentration}

\input{learning}

%\bibliographystyle{alpha}
%\bibliography{main}
\begingroup
\small
\printbibliography
\endgroup

\appendix

\end{document}

%% file: preamble.tex
\usepackage[T1]{fontenc}
\usepackage{lmodern}
\usepackage[margin=1in]{geometry}
\usepackage{microtype}
\usepackage{amsmath,amssymb,amsthm,mathtools,enumitem}
\usepackage[
  backend=biber,
  style=alphabetic,
  natbib=true,
  doi=true,
  eprint=true,
  url=false,
  maxbibnames=99,
  maxcitenames=99,
]{biblatex}
\usepackage{xcolor}
\usepackage{enumitem}
\usepackage{bm}
\usepackage{booktabs}
\usepackage{hyperref}
\usepackage{eucal}
\usepackage{thm-restate,mathrsfs} %To repeat theorems, lemmas etc....
\usepackage{comment} 
\usepackage{braket}
\usepackage{algorithm}
\usepackage{algpseudocode}
\usepackage{array,booktabs,tabularx}
\usepackage[font=footnotesize]{caption}
\usepackage{markdown}
\usepackage[most]{tcolorbox}
\usepackage[capitalise,noabbrev,nameinlink]{cleveref}
\usepackage{tikz}
\usetikzlibrary{arrows.meta}

\usepackage{titlesec}

\titleformat{\section}
  {\normalfont\Large\fontseries{b}\selectfont}
  {\thesection}{1em}{}

\titleformat{\subsection}
  {\normalfont\large\fontseries{b}\selectfont}
  {\thesubsection}{1em}{}

\hypersetup{
  colorlinks=true,
  linkcolor=blue!55!black,
  citecolor=green!40!black,
  urlcolor=blue!60!black
}

\allowdisplaybreaks
\setlist{leftmargin=*,topsep=3pt,itemsep=2pt}
\newtheorem{theorem}{Theorem}[section]

\newtheorem{lemma}[theorem]{Lemma}
\newtheorem{claim}[theorem]{Claim}

\theoremstyle{definition}

\numberwithin{equation}{section}
\numberwithin{figure}{section}
\numberwithin{table}{section}
\newcommand{\cE}{\mathcal{E}}

\newcommand{\cN}{\mathcal{N}}
\newcommand{\ketbra}[2]{\lvert #1\rangle\!\langle #2\rvert}

\newcommand{\norm}[1]{\lVert #1\rVert}

\newtcolorbox{algorithmbox}{
  enhanced,
  colback=gray!6,
  colframe=gray!35,
  boxrule=0.6pt,
  arc=2.5mm,
  width=0.94\textwidth,
  left=2mm,
  right=2mm,
  top=1.5mm,
  bottom=1.5mm,
  before skip=8pt,
  after skip=8pt,
}

\long\def\ca#1\cb{} %Use for commenting out: \ca...\cb

\newcommand{\poly}{\operatorname{poly}}

\newcommand{\fhat}{\widehat{f}}
\newcommand{\qhat}{\widehat{q}}
\newcommand{\bk}{\bm{k}}
\newcommand{\bz}{\bm{z}}
\newcommand{\Z}{\mathbb{Z}}
\newcommand{\U}{\mathrm{U}}

\def\01{\{0,1\}}
\newcommand{\eps}{\varepsilon}
\newcommand{\E}{\mathop{\mathbb{E}}}

\newcommand{\C}{{\mathbb C}}
\newcommand{\R}{\ensuremath{\mathbb{R}}}

\newcommand{\cL}{\mathcal{L}}

\newcommand{\rhodef}{\psi^{\mathrm{def}}}
\newcommand{\corr}{\mathrm{Corr}}

\newcommand{\haar}[1]{\mu_{\mathrm{Haar}(#1)}}

\newcommand{\Weyl}{\mu_{\mathrm{Weyl}}}

\renewcommand{\geq}{\geqslant}
\renewcommand{\leq}{\leqslant}

\makeatletter
\providecommand*{\cupdot}{%
  \mathbin{%
    \mathpalette\@cupdot{}%
  }%
}
\newcommand*{\@cupdot}[2]{%
  \ooalign{%
    $\m@th#1\cup$\cr
    \hidewidth$\m@th#1\cdot$\hidewidth
  }%
}
\makeatother

{}%         Body font
{}%         Indent amount (empty = no indent, \parindent 
\definecolor{lcInk}{HTML}{202939}
\definecolor{lcMuted}{HTML}{60718A}
\definecolor{lcGate}{HTML}{758298}
\definecolor{lcBlue}{HTML}{2563C7}
\definecolor{lcBlueFill}{HTML}{DCE9FD}
\definecolor{lcOrange}{HTML}{E67424}
\definecolor{lcOrangeFill}{HTML}{FCE6D2}
\definecolor{lcPurple}{HTML}{8548BF}
\definecolor{lcPurpleFill}{HTML}{D8C0EE}
\definecolor{gateblue}{RGB}{23,44,81}
\definecolor{BellPurple}{RGB}{133,72,191}
\definecolor{BellPurpleFill}{RGB}{232,217,248}
\definecolor{GqGreen}{RGB}{22,138,99}
\definecolor{GqGreenFill}{RGB}{217,242,232}

\allowdisplaybreaks

\newcommand{\ansbox}[1]{%
  \makebox[2.8em][c]{\fbox{\makebox[1.8em][l]{\strut #1}}}%
}
\newcommand{\yes}{\ansbox{Yes}}
\newcommand{\no}{\ansbox{No}}
\newcommand{\nostar}{\ansbox{No*}}

%% file: introduction.tex
\section{Introduction}

A central topic in theoretical computer science is understanding the interplay between
\emph{learnability} and \emph{pseudorandomness}~\cite{OS17}. These notions consider the same fundamental question from complementary directions: in learning, the goal is to recover structure in an object, whereas in pseudorandomness the goal is to conceal that structure. This interplay has played a pivotal role in complexity theory~\cite{CIKK16}, cryptography~\cite{BFKL93,KV94}, learning theory~\cite{DV21} and derandomization~\cite{OS18}. Recent work has begun to uncover analogous connections in the quantum setting~\cite{AGS21,CLS25, FGSY26}. A particularly concrete setting in which these two perspectives meet is the study of \emph{random quantum circuits}: given copies of a state prepared by a random circuit, one may ask whether the underlying circuit can be efficiently recovered, or instead whether its output is indistinguishable from a Haar-random state. The latter question leads to the notion of a \emph{pseudorandom quantum state} (PRS)~\cite{JLS18}: an ensemble of pure quantum states that can be efficiently prepared but remains indistinguishable from Haar random pure states to any efficient quantum adversary, even given polynomially many copies of the state. PRSs have many applications in quantum cryptography (see the references collected in~\cite{MicrocryptZoo}).

Throughout this work, we only consider local random circuits on fixed-dimensional brickwork architectures, where each local gate is sampled independently from the Haar measure and acts on two neighboring sites of the lattice. Although every gate acts only locally, repeated interactions can rapidly generate highly entangled and complex states~\cite{NRVH17}. This leads to two basic questions:

\vspace*{-5pt}
\begin{center}
\emph{When do local Haar random circuits on brickwork architectures become pseudorandom, and can they be efficiently learned up to the threshold of pseudorandomness?}
\end{center}
\vspace*{-5pt}

These questions remain largely open. We do not know how deep these circuits need to be to generate pseudorandom quantum states, how ancillary workspace affects the depth, or whether local Haar randomness suffices or structured cryptographic ensembles are required. Recent work~\cite{FGSY26} proposed the hardness of learning output states of random quantum circuits as a basis for quantum cryptography, providing further motivation. % for understanding the learning side of these questions.
Here, we focus on \emph{circuit depth and gate locality}. In one dimension with constant locality, many results point to logarithmic depth as the natural scale. A simple  argument shows that $O(\log n)$-depth circuits cannot generate PRSs~\cite{Pseudoentanglement}. By contrast,  structured phase construction of~\cite{CCGG}, which does not use  Haar random local gates, suggests that log-depth might be approached from above, assuming quantum-secure pseudorandom functions whose phase gates admit one dimensional implementations of linear depth. For local Haar random circuits, a recent work of LaRacuente~\cite{LaR26} provided complementary evidence by showing that for fixed $t$ and inverse polynomial error,  one dimensional Haar random brickwork circuits become approximate $t$-designs in nearly $\log n$ depth. Approximate designs do not imply computational pseudorandomness, but suggest that Haar-like behavior  emerges near~logarithmic~depth.

In higher dimensions, much less is known. Structured constructions show that higher connectivity and ancillary workspace can substantially reduce the required depth~\cite{SHH,CSBH25}. Without ancillas,  the best known cryptographic constructions remain at polylogarithmic depth~\cite{SMLBH25}, leaving open whether logarithmic depth is the ancilla-free threshold. Gate locality gives a second parameter. The same Schmidt rank argument described before extends to gates of locality $\ell$ and rules out PRSs in one dimension whenever $\ell d=O(\log n)$. Thus $\ell d$, rather than depth, is a natural combined scale for studying how far learning can extend towards pseudorandomness,~raising

\vspace*{-5pt}
\begin{center}
\emph{Can local random circuits on any fixed-dimensional brickwork architecture be learned in polynomial time from copies of their output states when $\ell d=O(\log n)$?}
\end{center}
\vspace*{-5pt}

\subsection{Main results}

Our main result answers the last question affirmatively. To describe our results, we first describe the setup. We sample a set of $M=\poly(n,d)$ unitaries $G_1,\ldots, G_M$ that act on $2\ell$ qubits from the Haar measure on $\U(4^\ell)$. Call this set $\cN$. For simplicity, we omit discussing issues related to precision in the description of these gates here and defer its discussion to later. We consider random $k$-dimensional brickwork circuits that are chosen by drawing each gate independently and uniformly from the fixed gateset\footnote{For constant locality, common choices for such problems include fixed universal gate sets~\cite{HL09,OSH20} and $\varepsilon$-nets~\cite{fefferman2024anticoncentrationunitaryhaarmeasure} (see also~\cite{KW} for learning with finite gate sets). At inverse polynomial precision, such nets have polynomial size for constant locality but size $\exp(n^{\Omega(1)})$ for $\Theta(\log n)$ locality. Compiling larger gates into a standard universal set of one- and two-qubit gates instead changes the depth and architecture. We therefore use a public gate set of polynomial size, retaining larger gates as individual operations without requiring coverage of all local unitaries. We note that making the gateset public does not make the problem substantially easier as the search space still consists of $2^{\poly(n)}$ candidate circuits. For $O(1)$ locality, our algorithm can also be adapted to work with a fixed $\varepsilon$-net.} $\cN$. After averaging over the random choice of $\cN$, the ideal circuit distribution is close in total variation distance to the ensemble with independent Haar gates, provided $M$ is sufficiently large (see \cref{sec:learning}). We show the following:

\begin{restatable}[\textbf{Learning Random Circuits with State Access}]{theorem}{learning}\label{thm:learning}
Let $d\geq 1$, $\delta>0$ and fix a constant dimension \(k\).  Let \(\cN\) be a random gateset as defined above.  With probability \(\geq 1-\delta\), we have: 
Let \(U\) be a depth-\(d\), \(k\)-dimensional brickwork circuit on \(n\) sites,
whose local gates are chosen uniformly from \(\cN\) where $d \le 0.01 n^{1/k}$.  Given copies of
$    |\psi_0\rangle = U|0^\ell\rangle^{\otimes n},$
there is an algorithm that, with probability \(\geq 1-\delta\) over
the random circuit and its measurement outcomes, exactly recovers \(U\), in sample and time complexity
\[
    \poly_k(n,d,\delta^{-1})\,2^{O_k(d\ell)}.
\]

\end{restatable}

In particular, the algorithm runs in polynomial time whenever $d\ell = O(\log n),$ which as the prior discussion highlighted seems to be a natural transition point to expect the emergence of pseudorandomness. While the guarantee of recovering the circuit $U$ exactly and only with state access may appear surprising, the proof in fact shows the stronger statement that with high probability two different random circuits above cannot generate states that are arbitrarily close. The key algorithmic insight is to identify gates using correlations between just two sites, without reconstructing the full lightcone (see \cref{sec:learning-overview}).

Our main technical contribution, which may be of independent interest, is an anticoncentration inequality for polynomials defined on Haar random unitaries. Informally, we show that a nonzero low-degree polynomial cannot be  close to zero on too large a subset (under the Haar measure) of the unitary group $\U(m)$.  The classical theorem of Carbery and Wright~\cite{CW01} gives such anticoncentration bounds for polynomials under the \emph{Gaussian measure}, and more generally any \emph{log-concave measure}. Extending their result  to the Haar measure is nontrivial primarily because  the unitary group is nonconvex and the matrix elements of unitaries are strongly correlated. In this direction, a recent work~\cite{fefferman2024anticoncentrationunitaryhaarmeasure} established anticoncentration for Haar-random unitaries by showing that for a degree-$d$ polynomial $P$ over $\U(m)$, we have
\begin{equation}\label{eqn:fgz}
    \Pr_{U\sim\haar{m}}
\left[|P(U)|\le \varepsilon\|P\|_{L^2(\haar{m})}\right]
\le C(m,d) \cdot \varepsilon^{1/(4m^2d)},
\end{equation}
where $C$ is a polynomial in $m$ and $d$. Furthermore, they showed that it was \emph{necessary} for $C$ to depend polynomially on $m$, but left open the question of whether the $m$-dependence in the small-ball exponent $\varepsilon$ can be removed or not, akin to how classic Carbery-Wright inequalities are dimension \emph{independent}. Resolving this question is the key bottleneck for handling random circuits with gates of growing locality.

We show a dimension-free small ball exponent is indeed possible. In particular, let us write $\|P\|_{\infty} = \sup_{U \in \U(m)} |P(U)|$ for a polynomial $P$. Then, we show the following\footnote{Since $\|P\|_{L^2(\haar{m})}\le\|P\|_\infty$, the event $|P(U)|\le\eps\|P\|_{L^2(\haar{m})}$ is contained in the event
$|P(U)|\le\eps\|P\|_\infty$. Thus our bound also implies the corresponding $L^2$-normalized bound as in \cref{eqn:fgz}}:

\begin{restatable}[\textbf{Carbery-Wright Inequality over the Unitary Group}]{theorem}{anticoncentration}
\label{thm:unitary-smallball}
Let $P: \U(m) \to \C$ be a non-zero complex-valued polynomial of total degree at most $d$ in the entries of $U$ and $\overline U$. 
Then, for every $\eps\in(0,1)$, the following holds
\[\Pr_{U \sim \haar{m}}\bigl[\,|P(U)| \le \eps \|P\|_{\infty} \bigr]
\le 64m^2 \cdot \eps^{\frac{1}{24\pi d}}.\]
\end{restatable}

The proof reduces the problem to a one dimensional \emph{Remez inequality} while keeping the exponent independent of dimension (see \cref{sec:anticoncentration-overview}). Anticoncentration inequalities have many applications in theoretical computer science, probability and mathematics~\cite{MNV16,NV13} and have become an important tool in quantum information.  They underlie complexity-theoretic arguments for the hardness of
Boson Sampling and random circuit sampling, and have also been used to study Haar-like
behavior, scrambling, unitary designs and learning of random quantum circuits
~\cite{BosonSampling,BFNW18,DHJB22,fefferman2024anticoncentrationunitaryhaarmeasure}.  We hope that the above dimension-free anticoncentration statement over the Haar measure will find other applications in quantum information and complexity. %be of independent interest. 

\subsection{Prior Work and Discussion}

We now give a thorough comparison with relevant prior works and explain previous bottlenecks. For an easy comparison with prior works, we also record the results in \cref{tab:comparison}. 
\begin{table}[!ht]
\centering
\captionsetup{font=footnotesize}
\small
\setlength{\tabcolsep}{3pt}
\renewcommand{\arraystretch}{1.2}
\begin{tabularx}{\linewidth}{@{}
  >{\raggedright\arraybackslash}p{0.14\linewidth}
  >{\centering\arraybackslash}p{0.115\linewidth}
  >{\centering\arraybackslash}p{0.085\linewidth}
  >{\raggedright\arraybackslash}p{0.10\linewidth}
  >{\raggedright\arraybackslash}p{0.125\linewidth}
  >{\centering\arraybackslash}p{0.16\linewidth}
  >{\raggedright\arraybackslash}X@{}}
\toprule
\textbf{Work}
& \textbf{Dimension}
& \textbf{Input locality}
& \textbf{Access}
& \textbf{Output}
& \textbf{Output depth}
& \textbf{Time} \\
\midrule
\multicolumn{7}{@{}l}{\emph{Arbitrary circuits}} \\
\addlinespace[3pt]

\cite{HuangEtAl24}
& 1D or 2D
& $O(1)$
& State copies
& Approximate state
& $3d$
& Quasipolynomial \\
\addlinespace[4pt]

\cite{HuangEtAl24}
& Fixed $k$
& $O(1)$
& Unitary queries
& Approximate $U$
& $O(d)$
& Quasipolynomial \\
\addlinespace[4pt]

\cite{LandauLiu25}\newline
\cite{KimKimRanard24}
& Fixed $k$
& $O(1)$
& State copies
& Approximate state
& $O(d)$
& Quasipolynomial \\
\addlinespace[4pt]

\cite{KimKimRanard24}
& 1D
& $O(1)$
& State copies
& Approximate state
& $2^{O(d)}$
& Polynomial \\
\addlinespace[4pt]

\cite{Hu}
& 1D
& $2\ell$
& State copies
& Approximate state$^\dagger$
& $\operatorname{poly}(n,2^{\ell d})$
& $\operatorname{poly}(n,2^{\ell d})$ \\
\midrule
\multicolumn{7}{@{}l}{\emph{Random brickwork circuits}} \\
\addlinespace[3pt]

\cite{fefferman2024anticoncentrationunitaryhaarmeasure}
& Fixed $k$
& $O(1)$
& Unitary queries
& Exact $U$
& $d$
& Polynomial \\
\addlinespace[4pt]

\textbf{Our work}
& Fixed $k$
& $2\ell$
& State copies
& Exact $U$
& $d$
& $\operatorname{poly}(n,2^{\ell d})$ \\
\bottomrule
\end{tabularx}

\smallskip
\begin{minipage}{\linewidth}
\footnotesize
Approximate reconstructions allow ancillas.
Exact recovery preserves the layout and uses the respective discrete
gate ensembles.
$^\dagger$\,Includes compilation of the learned channels into
preparation circuits.
\end{minipage}

\caption{Learning guarantees on lattices of fixed dimension $k$.
Unless shown explicitly, time is evaluated at $d=O(\log n)$ and
inverse polynomial error.}
\label{tab:comparison}
\end{table}

\textbf{Shallow Circuit Learning.} A recent line of work has studied the learnability of states prepared by shallow quantum circuits and, in turn, the reconstruction of circuits preparing these states. These works broadly consider two settings, worst-case and random circuits, under two different learning models. In the \emph{state access} model, one is given copies of a state $\ket{\psi}$ prepared by applying a shallow circuit to the all zero state, and the task is to recover a circuit that approximately prepares $\ket{\psi}$. This is the model considered in this paper. In the \emph{unitary access} model, one is given query access to the underlying circuit $U$, and the task is to learn a circuit that approximately implements $U$. For worst-case circuits, these two learning problems are incomparable. Unitary access provides strictly more information, but the learner is also required to recover the action of the entire unitary rather than only its action on the all zero state.

For worst-case circuits with unitary access, \cite{HuangEtAl24} show that constant-depth circuits can be learned in polynomial time in $n$. At logarithmic depth, however, the corresponding depth preserving algorithm runs in quasipolynomial time. For worst-case circuits with state access, the works~\cite{LandauLiu25,HuangEtAl24,KimKimRanard24} recover shallow circuits preparing the target state by learning sufficiently large local inversion channels. These results are \emph{remarkably general}, since they apply to worst case circuits, but their complexity is controlled by the size of the relevant lightcones, which becomes costly as the depth or the dimension of the lattice grows. Even for one dimensional brickwork circuits with $2$-local gates, the backward lightcone of a single output qubit contains $O(d)$ input qubits and $O(d^2)$ gate locations. A brute force search over gate assignments in this region can therefore require time $\exp(O(d^2))$, which is already superpolynomial for $d=\Theta(\log n)$. In $k$ dimensions, the number of gate locations in such a lightcone grows as $O_k(d^{k+1})$, making this dependence even more prohibitive. \cite{KimKimRanard24} and, more recently, \cite{Hu} obtain polynomial time learning in related regimes at the expense of substantially increasing the depth of the learned preparation. The latter work also allows growing gate locality, but only by outputting larger local quantum channels, whose circuit implementation may have substantially larger depth. % It therefore does not preserve the original locality or depth of the circuit.

For random circuits with unitary access, and perhaps the closest technical comparison to our work, is the work of Fefferman, Ghosh and Zhan~\cite{fefferman2024anticoncentrationunitaryhaarmeasure} who \emph{initiated} the study of learning random brickwork circuits. They considered learning $O(\log n)$-depth brickwork random circuits with unitary access and gave an exact learner that runs in polynomial time for $O(1)$ locality circuits. Their algorithm crucially relies on querying the unknown circuit on different inputs and observing their effect on selected output qubits. They analyze this using their unitary anticoncentration bound stated in \cref{eqn:fgz}. The dimension dependent exponent in that bound prevents the same argument from extending to gates of growing locality. Other recent works also study related learning problems in different architectures~\cite{KW, VasconcelosHuang24}.

Taken together, these results expose the main difficulty in moving from unitary access\footnote{We remark that the recent work~\cite{TangWrightZhandry25} showed how to simulate queries to a state preparation oracle using state copies. This simulation need not preserve the original unitary's action on other input states. Hence, this work cannot be combined with~\cite{fefferman2024anticoncentrationunitaryhaarmeasure} to obtain learning algorithms under state access, as~\cite{fefferman2024anticoncentrationunitaryhaarmeasure} require such guarantees.} to state access. With unitary access, the learner can actively probe the circuit by varying its input and observing how this change propagates to the output. With state access, all information must instead be extracted from copies of the single state $U{\ket{0^{\ell}}}^{\otimes n}$. The most direct approach is then to reconstruct suitable local reduced states, whose size is governed by backward lightcones. This makes going beyond the lightcone scale a central difficulty.

\textbf{MPS Learning Algorithms.}
Another approach is to exploit the tensor network structure of shallow circuit states. In one dimension, the output of a depth-$d$ circuit with $\ell$-local gates is an MPS with bond dimension $2^{O(\ell d)}$. Thus, when $\ell d=O(\log n)$, recent MPS learning algorithms can recover a succinct description of the output state in polynomial time from copies~\cite{bakshi2025learning,LinChiaHung25}. However, these methods recover an MPS description, or more generally a circuit preparing the learned MPS, rather than the original shallow circuit. In particular, they do not preserve the depth and local architecture of the unknown circuit. In our setting, for $O(1)$-local circuits, existing MPS learners~\cite{LinChiaHung25}, for instance, yield $O(n)$-depth implementations on one dimensional brickwork architectures. 

\textbf{Our work.} 
Putting together the comparisons of \cref{tab:comparison}, our work is the first to simultaneously obtain the following~guarantees:
\begin{enumerate}[label={$(\roman*)$}]
\item  runs in polynomial time for larger \emph{dimension and locality} as long as  $\ell d=O(\log n)$, 
\item requires only \emph{copies} of a single output state rather than unitary access, 

\item \emph{exactly recovers} the original circuit that prepared the state. Perhaps surprisingly, this much weaker state access model still suffices for exact circuit recovery for brickwork random  circuits.
\end{enumerate} The key insight is that the learner does not need to reconstruct an entire lightcone: carefully chosen correlations between only two output sites contain enough information to identify the gates layer by layer. In the next two sections, we give an overview of the proof, highlighting the main technical challenges and ideas.

\subsection{Overview of the Learning Algorithm}
\label{sec:learning-overview}

We describe our learning algorithm for one dimensional lattices here. The same ideas extend to higher-dimensional lattices. Let $U = U_d U_{d-1}\ldots U_1$ be a random circuit on one dimensional brickwork architecture chosen from the randomly chosen gateset $\cN$ described before the statement of \cref{thm:learning}. The learning algorithm will learn the layers of the circuit in backward fashion, first learning layer $U_d$, then $U_{d-1}$ and so on. All the layers will be learned exactly with high probability over the choice of the circuit and the gateset. Thus, we can always assume that for any $t \in [d]$, the algorithm can prepare copies of the state $\ket{\psi_t} := U_tU_{t-1}\ldots U_1\ket{0^\ell}^{\otimes n}$ prepared by the prefix circuit by applying the inverse circuit $U_{t+1}^\dag\ldots U^\dag_d$ to the copies of the input state $\ket{\psi_0} = U\ket{0^\ell}^{\otimes n}$. 

Let us assume that the algorithm is learning the layer $U_t$ and let $\psi = \ketbra{\psi}{\psi}$ be the state obtained by undoing the layers above. Consider a gate $G$ of the circuit which acts on two neighboring sites $B=i$ and $C=i+1$. Consider the site $D=i+2t-1$. Our starting point is the observation that the backward lightcones of the sites $B$ and $D$ are disjoint if the gate $G$ was removed from the circuit (see \cref{fig:lightcone-removal}). This implies that the state $\rhodef = G^\dag\psi G$ obtained by removing the gate $G$ from the circuit becomes separable when reduced to the sites $B$ and $D$. That is, $\rhodef_{BD} = \rhodef_{B}\otimes \rhodef_{D}$. 

\begin{figure}[H]
\centering

% ============================================================
% (a) Original circuit
% ============================================================
\begin{minipage}[t]{0.48\textwidth}
\centering
\resizebox{\linewidth}{!}{%
\begin{tikzpicture}[
    x=0.62cm,y=0.933cm,
    font=\sffamily\scriptsize,
    text=lcInk,
    line cap=round,
    line join=round,
    lc cone/.style={
        line width=0.35pt,
        dash pattern=on 1.5pt off 1.3pt,
        fill opacity=0.085
    },
    lc gate/.style={rounded corners=1.5pt}
]

\useasboundingbox (-0.8,-0.7) rectangle (15.8,6.0);

% B lightcone.
\path[lc cone,draw=lcBlue,fill=lcBlue]
  (2.875,4.861) -- (3.125,4.861)
  -- (3.125,4.347) -- (4.468,4.347)
  -- (4.468,3.347) -- (5.468,3.347)
  -- (5.468,2.347) -- (6.468,2.347)
  -- (6.468,1.347) -- (7.468,1.347)
  -- (7.468,0.389) -- (-0.468,0.389)
  -- (-0.468,1.347) -- (0.532,1.347)
  -- (0.532,2.347) -- (1.532,2.347)
  -- (1.532,3.347) -- (2.532,3.347)
  -- (2.532,4.347) -- (2.875,4.347) -- cycle;

% D lightcone.
\path[lc cone,draw=lcOrange,fill=lcOrange]
  (9.875,4.861) -- (10.125,4.861)
  -- (10.125,4.347) -- (10.468,4.347)
  -- (10.468,3.347) -- (11.468,3.347)
  -- (11.468,2.347) -- (12.468,2.347)
  -- (12.468,1.347) -- (13.468,1.347)
  -- (13.468,0.389) -- (5.532,0.389)
  -- (5.532,1.347) -- (6.532,1.347)
  -- (6.532,2.347) -- (7.532,2.347)
  -- (7.532,3.347) -- (8.532,3.347)
  -- (8.532,4.347) -- (9.875,4.347) -- cycle;

% Wires.
\foreach \q in {0,...,15}
  \draw[lcInk,line width=0.4pt]
    (\q,0.389) -- (\q,4.861);

% B-colored wire segments.
\def\BSegments{
  3/3/4.861/4.297,
  3/4/4.297/3.297,
  2/5/3.297/2.297,
  1/6/2.297/1.297,
  0/7/1.297/0.389
}

\foreach \lo/\hi/\yt/\yb in \BSegments {
  \foreach \q in {\lo,...,\hi}
    \draw[lcBlue,line width=0.7pt]
      (\q,\yb) -- (\q,\yt);
}

% D-colored wire segments.
\foreach \lo/\hi/\yt/\yb in {
  10/10/4.861/4.297,
  9/10/4.297/3.297,
  8/11/3.297/2.297,
  7/12/2.297/1.297,
  6/13/1.297/0.389
} {
  \foreach \q in {\lo,...,\hi}
    \draw[lcOrange,line width=0.7pt]
      (\q,\yb) -- (\q,\yt);
}

% Overlap.
\foreach \q in {6,7}
  \draw[lcPurple,line width=0.7pt]
    (\q,0.389) -- (\q,1.297);

% Brickwall gates.
\foreach \layer in {1,...,4} {

  \ifodd\layer
    \def\GateStarts{0,2,4,6,8,10,12,14}
  \else
    \def\GateStarts{1,3,5,7,9,11,13}
  \fi

  \foreach \a in \GateStarts {

    \pgfmathtruncatemacro{\InB}{
      (\a>=\layer-1) && (\a<=7-\layer)
    }

    \pgfmathtruncatemacro{\InD}{
      (\a>=\layer+5) && (\a<=13-\layer)
    }

    \def\GateDraw{lcGate}
    \def\GateFill{white}
    \def\GateWidth{0.4pt}

    \ifnum\InB=1\relax
      \def\GateDraw{lcBlue}
      \def\GateFill{lcBlueFill}
      \def\GateWidth{0.65pt}
    \fi

    \ifnum\InD=1\relax
      \def\GateDraw{lcOrange}
      \def\GateFill{lcOrangeFill}
      \def\GateWidth{0.65pt}

      \ifnum\InB=1\relax
        \def\GateDraw{lcPurple}
        \def\GateFill{lcPurpleFill}
      \fi
    \fi

    \draw[
      lc gate,
      draw=\GateDraw,
      fill=\GateFill,
      line width=\GateWidth
    ]
      ({\a-0.393},{\layer-0.297})
      rectangle
      ({\a+1.393},{\layer+0.297});

    \ifnum\InB=1\relax
      \ifnum\InD=1\relax
        \node[
          text=lcPurple,
          font=\sffamily\scriptsize\bfseries
        ]
        at ({\a+0.5},\layer) {shared};
      \fi
    \fi

    \ifnum\layer=4\relax
      \ifnum\a=3\relax
        \node[text=lcBlue] at (3.5,4) {$BC$};
      \fi
    \fi

  }
}

% Input dots.
\foreach \q in {0,...,15} {

  \def\DotColor{lcInk}

  \ifnum\q<6\relax
    \def\DotColor{lcBlue}
  \fi

  \ifnum\q>5\relax
    \ifnum\q<14\relax
      \def\DotColor{lcOrange}
    \fi
  \fi

  \ifnum\q>5\relax
    \ifnum\q<8\relax
      \def\DotColor{lcPurple}
    \fi
  \fi

  \fill[\DotColor]
    (\q,0.389) circle (2.5pt);
}

% Output wire labels.
\node[text=lcBlue,font=\large] at (3,5.083) {$B$};
\node[font=\large] at (4,5.083) {$C$};
\node[text=lcOrange,font=\large] at (10,5.083) {$D$};

% Legend at bottom.
\draw[draw=lcBlue,fill=lcBlueFill,line width=0.35pt]
  (1.0,-0.45) rectangle (1.4,-0.27);
\node[anchor=west,text=lcMuted]
  at (1.55,-0.36) {$B$ lightcone};

\draw[draw=lcOrange,fill=lcOrangeFill,line width=0.35pt]
  (5.3,-0.45) rectangle (5.7,-0.27);
\node[anchor=west,text=lcMuted]
  at (5.85,-0.36) {$D$ lightcone};

\draw[draw=lcPurple,fill=lcPurpleFill,line width=0.35pt]
  (9.6,-0.45) rectangle (10.0,-0.27);
\node[anchor=west,text=lcMuted]
  at (10.15,-0.36) {Overlap};

\end{tikzpicture}%
}

\smallskip

\scriptsize{(a) Original circuit}
\end{minipage}
\hfill
% ============================================================
% (b) Circuit after removing BC
% ============================================================
\begin{minipage}[t]{0.48\textwidth}
\centering
\resizebox{\linewidth}{!}{%
\begin{tikzpicture}[
    x=0.62cm,y=0.933cm,
    font=\sffamily\scriptsize,
    text=lcInk,
    line cap=round,
    line join=round,
    lc cone/.style={
        line width=0.35pt,
        dash pattern=on 1.5pt off 1.3pt,
        fill opacity=0.085
    },
    lc gate/.style={rounded corners=1.5pt}
]

\useasboundingbox (-0.8,-0.7) rectangle (15.8,6.0);

% B lightcone after removing BC.
\path[lc cone,draw=lcBlue,fill=lcBlue]
  (2.875,4.861) -- (3.125,4.861)
  -- (3.125,3.347) -- (3.468,3.347)
  -- (3.468,2.347) -- (4.468,2.347)
  -- (4.468,1.347) -- (5.468,1.347)
  -- (5.468,0.389) -- (-0.468,0.389)
  -- (-0.468,1.347) -- (0.532,1.347)
  -- (0.532,2.347) -- (1.532,2.347)
  -- (1.532,3.347) -- (2.875,3.347) -- cycle;

% D lightcone.
\path[lc cone,draw=lcOrange,fill=lcOrange]
  (9.875,4.861) -- (10.125,4.861)
  -- (10.125,4.347) -- (10.468,4.347)
  -- (10.468,3.347) -- (11.468,3.347)
  -- (11.468,2.347) -- (12.468,2.347)
  -- (12.468,1.347) -- (13.468,1.347)
  -- (13.468,0.389) -- (5.532,0.389)
  -- (5.532,1.347) -- (6.532,1.347)
  -- (6.532,2.347) -- (7.532,2.347)
  -- (7.532,3.347) -- (8.532,3.347)
  -- (8.532,4.347) -- (9.875,4.347) -- cycle;

% Wires.
\foreach \q in {0,...,15}
  \draw[lcInk,line width=0.4pt]
    (\q,0.389) -- (\q,4.861);

% B-colored wire segments.
\def\BSegments{
  3/3/4.861/4.297,
  3/3/4.297/3.297,
  2/3/3.297/2.297,
  1/4/2.297/1.297,
  0/5/1.297/0.389
}

\foreach \lo/\hi/\yt/\yb in \BSegments {
  \foreach \q in {\lo,...,\hi}
    \draw[lcBlue,line width=0.7pt]
      (\q,\yb) -- (\q,\yt);
}

% D-colored wire segments.
\foreach \lo/\hi/\yt/\yb in {
  10/10/4.861/4.297,
  9/10/4.297/3.297,
  8/11/3.297/2.297,
  7/12/2.297/1.297,
  6/13/1.297/0.389
} {
  \foreach \q in {\lo,...,\hi}
    \draw[lcOrange,line width=0.7pt]
      (\q,\yb) -- (\q,\yt);
}

% Brickwall gates.
\foreach \layer in {1,...,4} {

  \ifodd\layer
    \def\GateStarts{0,2,4,6,8,10,12,14}
  \else
    \def\GateStarts{1,3,5,7,9,11,13}
  \fi

  \foreach \a in \GateStarts {

    % Omit the BC gate in the fourth layer.
    \pgfmathtruncatemacro{\SkipGate}{
      (\layer==4) && (\a==3)
    }

    \ifnum\SkipGate=0\relax

      \pgfmathtruncatemacro{\InB}{
        (\a>=\layer-1) && (\a<=5-\layer)
      }

      \pgfmathtruncatemacro{\InD}{
        (\a>=\layer+5) && (\a<=13-\layer)
      }

      \def\GateDraw{lcGate}
      \def\GateFill{white}
      \def\GateWidth{0.4pt}

      \ifnum\InB=1\relax
        \def\GateDraw{lcBlue}
        \def\GateFill{lcBlueFill}
        \def\GateWidth{0.65pt}
      \fi

      \ifnum\InD=1\relax
        \def\GateDraw{lcOrange}
        \def\GateFill{lcOrangeFill}
        \def\GateWidth{0.65pt}

        \ifnum\InB=1\relax
          \def\GateDraw{lcPurple}
          \def\GateFill{lcPurpleFill}
        \fi
      \fi

      \draw[
        lc gate,
        draw=\GateDraw,
        fill=\GateFill,
        line width=\GateWidth
      ]
        ({\a-0.393},{\layer-0.297})
        rectangle
        ({\a+1.393},{\layer+0.297});

    \fi
  }
}

% Input dots.
\foreach \q in {0,...,15} {

  \def\DotColor{lcInk}

  \ifnum\q<6\relax
    \def\DotColor{lcBlue}
  \fi

  \ifnum\q>5\relax
    \ifnum\q<14\relax
      \def\DotColor{lcOrange}
    \fi
  \fi

  \fill[\DotColor]
    (\q,0.389) circle (2.5pt);
}

% Output wire labels.
\node[text=lcBlue,font=\large] at (3,5.083) {$B$};
\node[font=\large] at (4,5.083) {$C$};
\node[text=lcOrange,font=\large] at (10,5.083) {$D$};

% Legend at bottom.
\draw[draw=lcBlue,fill=lcBlueFill,line width=0.35pt]
  (2.4,-0.45) rectangle (2.8,-0.27);
\node[anchor=west,text=lcMuted]
  at (2.95,-0.36) {$B$ lightcone};

\draw[draw=lcOrange,fill=lcOrangeFill,line width=0.35pt]
  (7.0,-0.45) rectangle (7.4,-0.27);
\node[anchor=west,text=lcMuted]
  at (7.55,-0.36) {$D$ lightcone};

\end{tikzpicture}%
}

\smallskip

\scriptsize{(b) Circuit with $BC$ removed}
\end{minipage}

\caption{
Illustration of the learning algorithm on one-dimensional case. Here qubit $B$ and $D$ are two qudits with their lightcone just intersect.
(a) In the original brickwall circuit, the backward lightcones of
$B$ and $D$ intersect.
(b) After removing the gate acting on qudits $B, C$ at last layer, the backward lightcone of $B$ and $D$ doesn't intersect
}
\label{fig:lightcone-removal}

\end{figure}
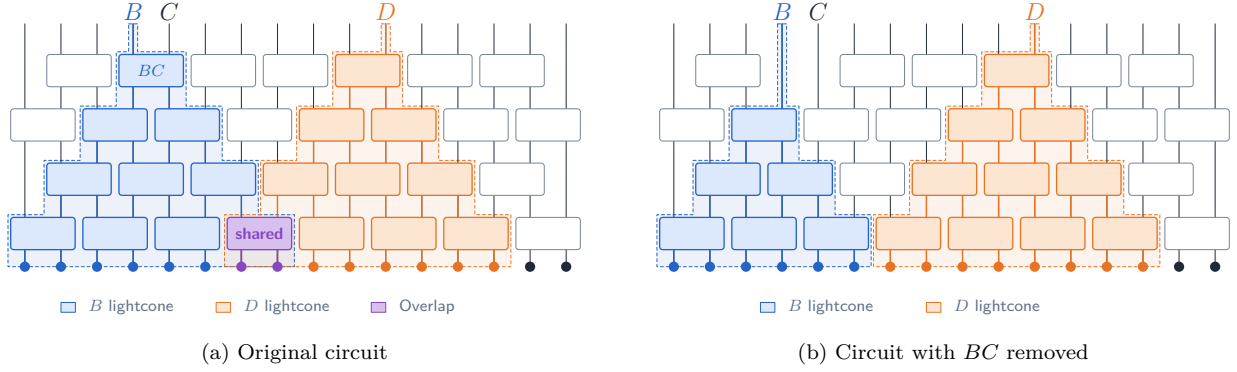

Define the following correlation function for an arbitrary quantum state $\rho$ to measure the distance from separability on two registers $B$ and $D$: 
\begin{align}\label{eqn:corr-intro}
        \corr_{B:D}(\rho) = \|\rho_{BD}-\rho_{B}\otimes \rho_D\|_F^2,
\end{align}
where the norm denotes the Frobenius norm.   This quantity measures the amount of correlation between $B$ and $D$ by comparing their joint state to the product of their marginals. It can be considered as some measure of mutual information between the sites $B$ and $D$ in the state $\rho$, so it vanishes when they are uncorrelated and becomes larger when the joint state is far from a product state.  In particular,  $\corr_{B:D}(\rhodef) =  0$ because of the disjointness of the backward lightcones.

Our key technical statement is that, with high probability, for the original state $|\psi\rangle$ in which the gate $G$ is present, this correlation is at least $2^{-O(\ell d)}$. Intuitively, the gate $G$ couples the two backward lightcones of $B$ and $D$ and thereby creates correlations between these regions. Our anticoncentration argument shows that, for a typical random circuit, $\corr_{B:D}(\rhodef)\geq 2^{-O(\ell d)}$. This scale is significantly larger than what one would obtain by paying for every gate in the full lightcone, which already contains $O(d^2)$ gates in one dimension and grows even faster in higher dimensions. A key part of our argument shows that, regardless of how all the other gates are fixed, one can choose only $O(d)$ gates so that the resulting correlation is constant. We then combine this observation with our Carbery-Wright anticoncentration inequality over the unitary group. The correlation can be written as a polynomial of constant degree in each gate. Since many independent gates are involved, we apply the anticoncentration bound inductively, exposing the relevant gates one at a time while controlling the loss at each step. This yields the desired high probability lower bound without degrading the small-ball exponent with the number of gates. The dimension-independent exponent is essential here for handling gates of larger than constant locality.

The above suggests a natural test to learn the gate $G$ exactly: apply the inverse of each gate in the public gateset to the state $\psi$; if one chooses the correct gate, the correlation becomes zero while for an incorrect guess, the correlation remains at least $2^{-O(\ell d)}$. This can be detected by performing tomography on the reduced system $BD$ with sufficient accuracy. In particular, since the dimension of the system $BD$ is $2^{O(\ell)}$ and the precision required is roughly $2^{-O(\ell d)}$, this implies that one such gate can be learned in time $\mathrm{poly}(n,\delta^{-1})2^{O(\ell d)}$. One can then proceed to learn the entire layer and continue with the next layer to learn the whole circuit.

\subsection{Overview of the Anticoncentration Proof}\label{sec:anticoncentration-overview}

The anticoncentration proof proceeds in three steps. %\snote{done with my pass + edits, looks great!}

\textbf{Reduction to Polytorus.} We first reduce the polynomial anticoncentration inequality on $\U(m)$ to a similar anticoncentration statement for \emph{Laurent polynomials} on the \emph{polytorus} $\U(1)^m$ where $\U(1) = \{z \in \C \mid |z|=1\}$ is the complex unit circle. Recall that a Laurent polynomial is simply a multivariate polynomial in which negative powers are also allowed. Thus, a Laurent polynomial $q: \U(1)^m \to \C$ can be written as 
\[ 
q(z_1,\ldots,z_m) = \sum_{(k_1,\ldots,k_m) \in \Z^m} \qhat(k_1,\ldots,k_m) z_1^{k_1}\cdots z_m^{k_m}, 
\]
where $\qhat(k_1,\ldots,k_m) \in \C$ and the number of non-zero terms is finite. The degree of a Laurent monomial $z^k$ is defined as the largest $\ell_1$-norm of any vector $(k_1,\ldots,k_m)$ that appears with a non-zero coefficient. We show that \cref{thm:unitary-smallball} reduces to showing that low-degree Laurent polynomials are anticoncentrated under a probability measure over the polytorus that we call the \emph{Weyl measure} denoted $\Weyl$. The {Weyl measure} captures how the eigenvalues of a Haar random unitary are distributed in a uniformly random order. Proving this relies on the Weyl integration formula and reduces our task to showing that
\[\Pr_{(z_1,\ldots,z_m)\sim \Weyl}\bigl[\,|q(z_1,\ldots,z_m)| \le \eps\|q\|_{\infty}\bigr] \lesssim m^2 \cdot \eps^{O(1/{d})}.\]

\textbf{Reduction to Univariate Extremal Inequalities.} To prove the above inequality, we use the idea of interpolation. Let $t \in [0,1]$ be a parameter that we refer to as time. We show that by considering the time parameter, one can reduce the above anticoncentration inequality over $\U(1)^m$ to a statement about extremal values of certain types of complex univariate functions of $t$. In particular, we consider functions of the form 
\[ h(t) = \sum_{\lambda \in [-\pi d,\pi d]} c_{\lambda} e^{i\lambda t},\]
where $\lambda$ takes only finite number of values and $\|h\|_{\infty} = h(0)$. Thus, the
anticoncentration problem reduces to understanding how large a set of
times $t$ can be on which $h$ is much smaller than $h(0)$.
In particular, if
\[
A=\{t: |h(t)|\le \varepsilon |h(0)|\},
\]
our goal is to upper bound the Lebesgue measure of $A\subseteq [0,1]$ in terms of $\varepsilon$ and the parameter $d$.

\textbf{Connection to Classical Remez Inequalities.} Understanding this question  is directly related to a generalization of extremal properties of Chebyshev polynomials known as \emph{Remez inequalities}. Suppose $h$ was a real-valued polynomial of degree $d$ that attains its maximum at $0$ and $A$ is an interval in $[-1,1]$. Then, 
\[ h(0) \le \left(\frac{1}{|A|}\right)^{O(d)} \sup_{x \in A} |h(x)|,\]
 since the largest growth of any polynomial that is bounded on the interval $A$ is attained by the Chebyshev polynomial of degree $d$ rescaled on the interval $A$ (see, e.g.,~\cite{sachdeva2013approximation}). The above implies that if $A$ is an interval where $|h(x)| \le \eps |h(0)|$, then $|A| \lesssim \varepsilon^{O(1/d)}$. Remez inequalities generalize this statement to arbitrary measurable sets $A$ that need not be intervals. In our setting, $h$ is not a polynomial but an exponential sum whose 
frequencies lie in a bounded interval, so the relevant notion of complexity is
its \emph{bandwidth}.  There is an extensive literature on Remez inequalities for other families of functions with the role of degree replaced by some other property and we derive the statement needed for our application from a work of  Nazarov, Sodin and Volberg~\cite{NSV} with a short proof.

%% file: preliminaries.tex
\section{Preliminaries}

\subsection{Notation} We write $[k]=\{1,\ldots,k\}$. We use boldfaced letters $\bz=(z_1,\ldots,z_m)$ to denote tuples or vectors. For a pure quantum state $\ket{\psi}$, we write $\psi = \ketbra{\psi}{\psi}$ to denote the corresponding density matrix. We write $\psi_{AB}$ to denote the reduced density matrix on the subsystems $A$ and $B$. We will sometimes use $\|\cdot\|_{\infty}$ and $\|\cdot\|_2$ to denote the $L_\infty$ and $L_2$ norms of a function with respect to the relevant underlying measure. The measure may vary from one setting to another and will be clear from context. We write $\bm{1}[\mathcal{E}]$ to denote the indicator function of an event $\mathcal{E}$.% \mnote{Put other basic things that are needed here. The architecture can be defined with the problem}

We use $\U(m)$ to denote the group of $m \times m$ unitary matrices and $\haar{m}$ to denote the normalized Haar measure over $\U(m)$. We refer to the group $\U(1)$ as the torus and the direct product $\U(1)^m$ as the polytorus. The $\haar{1}^{\otimes m}$ denotes the product probability measure where each coordinate is sampled from the Haar measure on $\U(1)$, i.e., from the uniform measure on the complex unit circle.

\subsection{Weyl Measure and Weyl Integration}
The \emph{Weyl probability measure} captures how the eigenvalues of a Haar random unitary are distributed in a uniformly random order and we briefly describe some of its properties. The {Weyl} measure $\Weyl$ over $\U(1)^m$ is absolutely continuous with respect to the product measure $\haar{1}^{\otimes m}$ on the polytorus and defined by the following relative density for every $\bz \in \U(1)^m$,
\begin{equation}\label{eqn:weyl}
     \frac{d\Weyl}{d\haar{1}^{\otimes m}}(\bz) = \frac{1}{m!}\prod_{j < k} |z_j - z_k|^2.
\end{equation}

The density is always non-negative and Lemma \ref{lem:weyl-normalization} below shows that it is indeed a probability measure. 

\begin{lemma}
\label{lem:weyl-normalization}
 $\E_{\bz \sim \haar{1}^{\otimes m}}\left[\frac{1}{m!} \prod_{i < j} |z_i - z_j|^2\right] = 1.$
\end{lemma}

\begin{proof}
Consider the Vandermonde determinant
\[
\Delta(\bz)
=
\det
\begin{pmatrix}
z_1^{m-1} & z_1^{m-2} & \cdots & z_1 & 1 \\
z_2^{m-1} & z_2^{m-2} & \cdots & z_2 & 1 \\
\vdots   & \vdots   &        & \vdots & \vdots \\
z_m^{m-1} & z_m^{m-2} & \cdots & z_m & 1
\end{pmatrix}
=
\prod_{i<j}(z_i-z_j).
\]
Expanding the determinant over permutations gives
% \[
% \Delta(\bz)
% =
% \sum_{\pi\in S_m}
% \operatorname{sgn}(\pi)
% \prod_{j=1}^m z_j^{\,m-\pi(j)}.
% \]
% Thus, 
\begin{align*}
|\Delta(\bz)|^2
&=
\sum_{\pi,\sigma\in S_m}
\operatorname{sgn}(\pi)\operatorname{sgn}(\sigma)
\prod_{j=1}^m
z_j^{\,m-\pi(j)}
\overline{z_j^{\,m-\sigma(j)}} =
\sum_{\pi,\sigma\in S_m}
\operatorname{sgn}(\pi)\operatorname{sgn}(\sigma)
\prod_{j=1}^m z_j^{\,\sigma(j)-\pi(j)}.
\end{align*}

\noindent Taking expectation over $\bz$, all the terms where $\sigma \neq \pi$ vanish. Thus, 
\begin{align*}
\E_{\bz \sim \haar{1}^{\otimes m}}|\Delta(z)|^2  
&=
\sum_{\pi,\sigma\in S_m}
\operatorname{sgn}(\pi)\operatorname{sgn}(\sigma)
\mathbf{1}[\pi=\sigma] = m!\qedhere
\end{align*}
\end{proof}

A function $f: \U(m) \to \C$ is called a \emph{class function} for the unitary group $\U(m)$ if $f(V) = f(SVS^*)$ for any unitary $S$. The Weyl integration formula gives a way to compute the Haar average of a class function by expressing it in terms of an average over the Weyl measure.

\begin{theorem}[Weyl Integration Formula \cite{Z24}, Proposition B.12]\label{lem:weyl-integration} For any bounded measurable class function $f: \U(m) \to \C$, we have
     $$\E_{V \sim \haar{m}}[f(V)] = \E_{\bz \sim \Weyl}[f(\mathrm{Diag}(\bz))].$$
\end{theorem}

%% file: anticoncentration.tex
\section{Carbery-Wright Inequality over the Unitary Group}

 For a function $f: \U(m) \to \C$, we write $$\|f\|_2 = \|f\|_{L^2(\haar{m})}= \left(\E[|f(U)|^2]\right)^{1/2}$$ and write $\|f\|_{\infty} = \sup_{U \in \U(m)} |f(U)|$ below. We prove the following Carbery-Wright type anticoncentration inequality over the unitary group. 

\anticoncentration*

The proof proceeds in three steps. First, we reduce the polynomial anticoncentration inequality on $\U(m)$ to proving a similar anticoncentration type inequality on the polytorus $\U(1)^m$. This inequality concerns \emph{Laurent polynomials} defined below but under the $\Weyl$ probability measure over $\U(1)^m$. 

\paragraph{Anticoncentration over the Polytorus.}

Let $\bz = (z_1,\ldots, z_m) \in \U(1)^m$ and $\bk \in \Z^m$, let us write $\bz^{\bk} = z_1^{k_1}\cdots z_m^{k_m}$. A function $f: \U(1)^m \to \C$ is called a Laurent polynomial if 
\[ f(\bz) = \sum_{\bk \in \Z^m} \fhat(k) \bz^{\bk}, \]
where $\fhat(k) \in \C$ and the number of non-zero terms is finite. The degree of a Laurent monomial $z^k$ is defined as $|k_1| + \ldots + |k_m|$ and the total degree of $f$ is defined to be the largest degree of a non-zero Laurent monomial. Note that $k_i$ may take negative values as well. We denote $\|f\|_{\infty} = \sup_{z \in \U(1)^m} |f(z)|$. We prove the following anticoncentration statement for Laurent polynomials with respect to the Weyl measure.

\begin{theorem}
\label{thm:polytorus-smallball}
Let $q: \U(1)^m \to \C$ be a non-zero Laurent polynomial of total degree at most $d$ such that $\|q\|_{\infty} = q(\bm{1})$. Then, for every $\eps\in(0,1)$, the following holds
\[\Pr_{\bz\sim \Weyl}\bigl[\,|q(\bz)| \le \eps\|q\|_{\infty}\bigr] \le 64m^2 \cdot \eps^{\frac{1}{24\pi d}}.\]
\end{theorem}

To prove the above inequality, we use the idea of interpolation. Let $t \in [0,1]$ be a parameter that we refer to as time. We show that by considering the time parameter, one can reduce the above anticoncentration inequality over $\U(1)^m$ to a statement about extremal values of certain types of complex univariate functions of $t$. In particular, we consider functions of the form 
\[ h(t) = \sum_{\lambda \in \Lambda} c_{\lambda} e^{i\lambda t}\]
where $\Lambda$ is a finite set in the bounded interval $[-\pi d,\pi d]$, where the interval is determined by  the degree $d$ of the Laurent polynomial.  Furthermore, the other assumption on $q$ will imply that the function $h$ attains its maximal value at $t=0$. Our goal will be reduced to bounding the Lebesgue measure of a set $A \subseteq [0,1]$ where $h$ takes much smaller values compared to $h(0)$.

\paragraph{The Univariate Extremal Statement.} For a measurable set $A \subseteq \R$, we write $|A|$ to denote the Lebesgue measure of $A$. We prove the following statement about functions of the above form.
\begin{theorem}\label{thm:remez}
Let $h: \R \to \C$ be a non-zero function of the form $$
h(t) = \sum_{\lambda \in \Lambda} c_{\lambda} e^{i\lambda t}$$ where $\Lambda$ is a finite set in some bounded interval $[-\tau,\tau]$. Suppose $\sup_{t \in \R} |h(t)| = h(0)$. Then, for every measurable $A \subseteq [-1,1]$ with non-zero measure $|A|$, we have 
\[ h(0) \le \left(\frac{32}{|A|}\right)^{24\tau} \cdot  \sup_{x \in A} |h(x)|.\]
\end{theorem}

Note that if $A$ is the set where $|h(x)| \le \eps |h(0)|$, then rearranging the above, we obtain that $|A| \le 32 \varepsilon^{1/(24\tau)}.$
  One can understand the above statement by relating it to the extremal properties of Chebyshev polynomials (see, e.g.,~\cite{sachdeva2013approximation}). For intuition, first consider the case where $h$ is a real-valued polynomial of degree $d$ that attains its maximum at $0$ and $A$ is an interval in $[-1,1]$. Then, 
 \[ h(0) \le \left(\frac{32}{|A|}\right)^{O(d)} \sup_{x \in A} |h(x)|,\]
 since the largest growth of any polynomial that is bounded on the interval $A$ is attained by the Chbeyshev polynomial of degree $d$ rescaled on the interval $A$. For us, there are two main differences. First, $A$ is an arbitrary measurable set instead of an interval. It turns out that the extremal properties of Chebyshev polynomials extend to this case and such inequalities are referred to as Remez inequalities in the literature and proven for various families of functions with the role of degree replaced by some other property. Second, for us $h$ is a sum of exponential functions and we want $\tau$ to be a proxy for the degree. While such an inequality is not known in the literature, we can derive the required statement from a short proof using a theorem of Nazarov, Sodin and Volberg~\cite{NSV} (see \cref{sec:remez}).

\begin{proof}[Proof of \cref{thm:unitary-smallball}]
   The proof proceeds in three steps. In \cref{sec:torus-reduction},
we prove \cref{thm:unitary-smallball} assuming
\cref{thm:polytorus-smallball}. In \cref{sec:remez-reduction},
we prove \cref{thm:polytorus-smallball} assuming
\cref{thm:remez}. Finally, we prove \cref{thm:remez}
in \cref{sec:remez}.
\end{proof}

\subsection{Reduction to the Polytorus}\label{sec:torus-reduction}

We first show that \cref{thm:polytorus-smallball} implies \cref{thm:unitary-smallball}. Let $U_0 \in \U(m)$ be a unitary where the polynomial $P$ takes its maximum modulus. First, we claim that
\begin{align}\label{eqn:P}
\Pr_{U \sim \haar{m}}
\bigl[\,|P(U)| \le \eps \|P\|_{\infty} \bigr]
&=
\Pr_{\substack{W \sim \haar{m}\\ \bz \sim \Weyl}}
\bigl[\,|P(U_0W\mathrm{Diag}(\bm{z})W^*)|
       \le \eps \|P\|_{\infty} \bigr],
\end{align}
where $W$ and $\bz$ are independent. To see the above, we define the function $f: \U(m) \to \mathbb{R}$ as 
\[ f(V) = \E_{W \sim \haar{m}}\bm{1}\bigl[\,|P(U_0WVW^*)| \le \eps \|P\|_{\infty} \bigr].\]
The function $f$ is a bounded class function, i.e. $f(V) = f(SVS^*)$ for any unitary $S$ because of the invariance of the Haar measure. Thus, the Weyl integration formula (\cref{lem:weyl-integration}) implies that 
\begin{align*}
     \E_{V \sim \haar{m}}[f(V)] &= \E_{\bz \sim \Weyl}[f(\mathrm{Diag}(\bz))]\\
                                &= \Pr_{\substack{W \sim \haar{m}\\ \bz \sim \Weyl}}
\bigl[\,|P(U_0W\mathrm{Diag}(\bm{z})W^*)|
       \le \eps \|P\|_{\infty} \bigr],
\end{align*}
proving \cref{eqn:P}. Based on \cref{eqn:P}, for a unitary $W \in \U(m)$, we define a function $q_W: \U(1)^m \to \C$  on the polytorus as follows $q_W(\bm{z}) = p(U_0 W \mathrm{Diag}(\bm{z}) W^*)$. To prove \cref{thm:unitary-smallball},  it suffices to bound $\Pr\bigl[\,|q_W(\bz)| \le \eps \|q_W\|_{\infty} \bigr]$ where $W \sim \haar{m}$ and $\bz \sim \Weyl$. Note that for any unitary $W$, the following two conditions hold: the function $q_W$ is a non-zero Laurent polynomial on the polytorus with total degree $d$ and $q_W(\bm{z})$ attains its maximum value at $\bm{z}=\bm{1}$. Thus, we can apply \cref{thm:polytorus-smallball} to $q_W$ obtaining
\[ \Pr_{\bz\sim \Weyl}\bigl[\,|q_W(\bz)| \le \eps\|q_W\|_{\infty}\bigr] \le 64m^2 \cdot \eps^{\frac{1}{24\pi d}}.\]
Averaging over $W \sim \haar{m}$ gives the statement of \cref{thm:unitary-smallball}.

\subsection{Reduction to Univariate Extremal Inequalities} \label{sec:remez-reduction}

Next, we show that \cref{thm:polytorus-smallball} is implied by  \cref{thm:remez}. Throughout, this proof $\bm{z}$ is distributed according to $\Weyl$ unless explicitly stated otherwise. We start by recalling that $q: \U(1)^m \to \C$ is a non-zero Laurent polynomial of total degree $d$ and $\|q\|_\infty = q(\bm{1})$. Let $\bm{1}_E(\bz)$ denote the indicator function of the event $|q(\bm{z})| \le \varepsilon |q(\bm{1})|$. We want to bound $\E[\bm{1}_E(\bm{z})]$. To do this, we introduce an interpolation parameter $t \in [0,1]$ and define the univariate function $f(t) = \E\left[\bm{1}_E(\bm{z}) \cdot \bm{1}\left[\max_j |\mathrm{arg}(z_j)| \le \pi t\right]\right]$, noting that $f(0)=0$ and $f(1)=\E[\bm{1}_E(\bm{z})]$.

By definition, $\bm{z} \sim \Weyl$ can be generated by first sampling a $\bz$ from the product measure $\haar{1}^{\otimes m}$ and then reweighting it according to \cref{eqn:weyl}. With this in mind, for $\bm{V} \in [-1,1]^m$, define 
\[ \bm{z}_{\bm{V}}(t)=(e^{i\pi t V_1/\|\bm{V}\|_{\infty}},\ldots,e^{i\pi t V_m/\|\bm{V}\|_{\infty}}).\]
Let us use the shorthand $\bm{z}_{V_i}(t) = e^{i\pi t V_i/\|\bm{V}\|_{\infty}}$ to denote the $i$th coordinate above. 

If $\mathrm{Unif}[-1,1]^m$ denotes the uniform probability measure on the cube $[-1,1]^m$, then one can check that the random variable $V/\|V\|_{\infty}$ is independent of $\|V\|_{\infty}$. Thus, $\bm{z}_{\bm{V}}(\|V\|_{\infty})= (e^{i\pi V_1},\ldots,e^{i\pi V_m})$ is distributed as the product measure $\haar{1}^{\otimes m}$. We make the following technical claim with $t$ denoting the value taken by $\|V\|_{\infty}$.

\begin{claim}
\label{claim:weights}
There exists non-negative weights $\psi_{\bm{V}}:[0,1]\to\R_{\ge 0}$ satisfying 
\begin{enumerate}[leftmargin=*,itemsep=1pt,topsep=2pt]
    \item $\max_{t \in [0,1]} \psi_{\bm{V}}(t) \le 2m^2 \int_0^1 \psi_{\bm{V}}(s)ds,$
    \item $\E_{\bm{V} \sim \mathrm{Unif}[-1,1]^m}\left[\int_0^1\psi_{\bm{V}}(t)dt\right]=1$,
\end{enumerate}
such that the derivative of $f$ with respect to $t$ is given by
\begin{equation}\label{eqn:fprime}
    \ f'(t) = \E_{\bm{V} \sim \mathrm{Unif}[-1,1]^m}\left[\bm{1}_E(\bm{z}_{\bm{V}}(t)) \cdot \psi_{\bm{V}}(t)\right].
\end{equation}
\end{claim}
We prove the above claim later and first finish the proof of \cref{thm:polytorus-smallball}. 

For every $\bm{V} \in [-1,1]^m$, consider the function $h_{\bm{V}}:\R \to \C$ defined as $h_{\bm{V}}(t) = q(\bm{z}_{\bm{V}}(t)) = \sum_{\|\bm{k}\|_1\le d} \widehat{q}(k) e^{i\pi t \langle\bm{k},\bm{V}/\|\bm{V}\|_{\infty}\rangle}$. Since $|\langle\bm{k},\bm{V}/\|\bm{V}\|_{\infty}\rangle| \le d$ above, the function $h_{\bm{V}}$ is an exponential function of the form in \cref{thm:remez} with $\tau=\pi d$. Furthermore, $h_{\bm{V}}(0)=q(\bm{1})$ and $|h_{\bm{V}}(t)|\le |q(\bm{1})|$ for every $t \in \R$. This implies that $A_{\bm{V}}  = \{ t \subseteq [0,1] \mid \bm{z}_{\bm{V}}(t) \in E\} = \{ t \subseteq [0,1] \mid |h_{\bm{V}}(t)| \le \varepsilon |h_{\bm{V}}(0)|\}$ and therefore $\int_0^1 \bm{1}_E(\bm{z}_{\bm{V}}(t)) dt = |A_{\bm{V}}|$.

%\cref{thm:remez} implies that $|A_{\bm{V}}| \le 100\varepsilon^{1/40d}$. 
By \cref{thm:remez},
$$
    |h(0)|
    \le
    \left(\frac{32}{|A|}\right)^{24\tau}
    \sup_{t\in A}|h(t)|.
$$
Applying this to $h_{\bm V}$ and $A_{\bm V}$ gives
$$
    |h_{\bm V}(0)|
    \le
    \left(\frac{32}{|A_{\bm V}|}\right)^{24\tau}
    \varepsilon |h_{\bm V}(0)|.
\implies    1
    \le
    \left(\frac{32}{|A_{\bm V}|}\right)^{24\tau}\varepsilon.
$$
and this implies $    |A_{\bm V}|
    \le
    32\,\varepsilon^{1/(24\tau)}.$ 
In our application, the frequencies of $h_{\bm V}$ lie in
$[-\pi d,\pi d]$, so we may take $\tau=\pi d$ and conclude that $   |A_{\bm V}|
    \le
    32\,\varepsilon^{1/(24\pi d)}.$ Thus, using \cref{claim:weights} and the fundamental theorem of calculus, we have that 
\begin{align*}
    f(1) &=  \int_0^1 \E_{\bm{V} \sim \mathrm{Unif}[-1,1]^m}\left[\bm{1}_E(\bm{z}_{\bm{V}}(t)) \cdot \psi_{\bm{V}}(t)\right]dt\\
     &\le 2m^2  \E_{\bm{V} \sim \mathrm{Unif}[-1,1]^m}\left[ \int_0^1\bm{1}_E(\bm{z}_{\bm{V}}(t)) \cdot \left(\int_0^1\psi_{\bm{V}}(s)ds\right)dt\right],
\intertext{
where we replaced the non-negative weight $\psi_{\bm{V}}(t)$ by the upper bound from the claim and used Fubini's theorem to switch the expectation and the integral. Since the second integral does not depend on $t$, we have}
    f(1) &\le 2m^2  \E_{\bm{V} \sim \mathrm{Unif}[-1,1]^m}\left[ |A_{\bm{V}}|\cdot \left(\int_0^1\psi_{\bm{V}}(s)ds\right)\right]\\
     &\le 64m^2 \varepsilon^{1/(24\pi d)} \cdot \E_{\bm{V} \sim \mathrm{Unif}[-1,1]^m}\left[\int_0^1\psi_{\bm{V}}(s)ds\right] \\
     & = 64m^2 \varepsilon^{1/(24\pi d)},
\end{align*}
where the last equality used the second property of the weights $\psi_{\bm{V}}$ from the claim.

\begin{proof}[Proof of \cref{claim:weights}]
     For this claim, it will be convenient to prove the lemma assuming $E \subseteq \U(1)^m$ is an arbitrary Borel set. Even in this setting the identity $f(1)= \E_{\bm{z}\sim\Weyl}[\bm{1}_E(\bm{z})]$ established earlier continues to hold.
     
For $\bm{V} \sim \mathrm{Unif}[-1,1]^m$, the random variable $\|\bm{V}\|_{\infty}$ is $\mathrm{Beta}(m,1)$ distributed with the probability density function $mt^{m-1}$ at $t$ and is independent of $\bm{V}/\|\bm{V}\|_{\infty}$, so we have
\[ f(t) = \int_0^t\E_{\bm{V} \sim \mathrm{Unif}[-1,1]^m}\left[\bm{1}_E(\bm{z}_{\bm{V}}(s)) \cdot ms^{m-1} \cdot \frac{1}{m!}{\prod_{i<j} \left|{\bm{z}}_{V_i}(s) -{\bm{z}}_{V_j}(s)\right|^2}\right]ds.\]

We define $\psi_{\bm{V}}(t) = \frac{t^{m-1}}{(m-1)!} {\prod_{i<j} \left|{\bm{z}}_{V_i}(t) -{\bm{z}}_{V_j}(t)\right|^2}$ and the identity \eqref{eqn:fprime} in the statement of the claim follows by taking the derivative with respect to $t$. The normalization $\E_{V \sim \mathrm{Unif}[-1,1]^m}\left[\int_0^1\psi_{\bm{V}}(t)dt\right]=1$ follows by recalling that $f(1)= \E_{\bm{z}\sim\Weyl}[\bm{1}_E(\bm{z})]$ and taking $E=\U(1)^m$. 

It remains to establish the first property of the weights $\psi_{\bm{V}}$. Towards this, note that $|e^{ix}-e^{iy}|^2 = 4 \sin^2\left(\frac{x-y}{2}\right)$. 
Defining $\beta_{ij} = \pi(V_i-V_j)/(2\|\bm{V}\|_{\infty})$, we note that $\beta_{ij} \in [0,\pi]$. We have $\psi_{\bm{V}}(t) = \frac{t^{m-1}}{(m-1)!} \prod_{i < j} 4\sin^2(\beta_{ij} t)$. In the interval $t \in [0,1/2]$, both functions $t^{m-1}$ and $\sin(\beta_{ij}t)$ are non-decreasing, so the maximum value of $\psi_{\bm{V}}(t)$ is attained at $t_{\max} \in [1/2,1]$.  Concavity of the sine function in the interval $[0,\pi]$ implies that $\sin(\beta_{ij}t) \ge \frac{t}{t_{\max}} \sin(\beta_{ij} t_{\max})$.

Squaring and taking the product of all distinct pairs, we obtain 
\begin{align*}
    \psi_{\bm{V}}(t) &\ge \frac{t^{m-1}}{(m-1)!} \left(\frac{t}{t_{\max}}\right)^{2\binom{m}{2}} \prod_{i < j} 4\sin^2(\beta_{ij} t_{\max})\\
    &= \left(\frac{t}{t_{\max}}\right)^{m-1} \left(\frac{t}{t_{\max}}\right)^{2\binom{m}{2}} \cdot \frac{t_{\max}^{m-1}}{(m-1)!}\prod_{i < j} 4\sin^2(\beta_{ij} t_{\max})\\
     & =\left(\frac{t}{t_{\max}}\right)^{m^2-1}   \psi_{\bm{V}}(t_{\max}).
\intertext{
Integrating both sides}
    \int_0^1\psi_{\bm{V}}(t)dt &\ge \psi_{\bm{V}}(t_{\max}) \int_0^{t_{\max}}\left(\frac{t}{t_{\max}}\right)^{m^2-1}dt\\
    &= \psi_{\bm{V}}(t_{\max}) \cdot t_{\max} \int_0^{1}s^{m^2-1}ds\\
    &= \psi_{\bm{V}}(t_{\max}) \cdot t_{\max} \cdot \frac{1}{m^2}.
\end{align*}
Using that $t_{\max}\ge 1/2$ implies the first property. This completes the proof of the claim.\qedhere

\end{proof}

\subsection{Proof of the Univariate Extremal Inequality}\label{sec:remez}

Now, we prove \cref{thm:remez} which will complete the proof. To prove this theorem, we will need the following theorem by Nazarov, Sodin and Volberg~\cite{NSV}.
%\mnote{I am merging this section with the other notation}

\begin{lemma}[{\cite[Lemma B]{NSV}}]\label{lem:remez}
Let $f$ be analytic in $D=\{z\in\C:|z|<1\}$ and suppose for some $a\in (0,1)$ we have
\[
 \sup_{z\in D}|f(z)|\leq1,
 \qquad |f(-a)|=|f(a)|=:b>0.
\]
Then for every sub-interval $I\subset[-a,a]$ and every set $E\subseteq I$, we have
\[
 \max_{x\in I}|f(x)|
 \leq\left(\frac{8|I|}{|E|}\right)^{\sigma}
 \sup_{x\in E}|f(x)|,
\]
where $\sigma=3/(1-a)\log 1/b$
\end{lemma}

\begin{proof}[Proof of \cref{thm:remez}]
Let $M=h(0)$. For simplicity we will work with the polynomial $H(z)=h(z)/M$, so that $|H(0)|=1$
and $|H(x)|\leq1$ for all real $x$. To use \cref{lem:remez}, we need a function $f$ that is bounded on the complex unit disk and an interval $I \subseteq [-a,a]$ where $a$ and $-a$ are points where the function $f$ has the same non-zero magnitude. 

Towards this, we first bound the value of $H$ on the complex unit disk. We show that
\begin{equation}
\label{eq:growth}
|H(x+iy)|\le e^{\tau |y|}
\end{equation}
for $x,y\in \R$. 
Consider first the upper half-plane and define
$$
    F(z):=e^{i\tau z}H(z)
    =
    \frac1M\sum_{\lambda\in\Lambda}
    c_\lambda e^{i(\lambda+\tau)z}.
$$
Since \(\lambda+\tau\ge0\) for every \(\lambda\in\Lambda\), each summand
is bounded in the upper half-plane, and hence \(F\) is bounded there.
Moreover, on the real axis,
$    |F(x)|=|H(x)|\le1.$
 Thus the hypotheses of the Phragmen-Lindelof principle
\cite[Lemma~9]{ACL} are satisfied (after reflecting the lower-half-plane
form stated there to the upper half-plane), and therefore
$
    |F(z)|\le1
 $ for \(\operatorname{Im}z\ge 0\).
For \(z=x+iy\) with \(y\ge0\),
$    |F(z)|
    =
    e^{-\tau y}|H(x+iy)|,$
so Eq.~\eqref{eq:growth} follows.
For the lower half-plane, apply the same argument to
$$
    F_-(z):=e^{-i\tau z}H(z)
    =
    \frac1M\sum_{\lambda\in\Lambda}
    c_\lambda e^{i(\lambda-\tau)z}.
$$
Since \(\lambda-\tau\le0\), this function is bounded in the lower half-plane, and the same Phragmen-Lindelof principle gives $    |H(x+iy)|\le e^{-\tau y}$ for $y\leq 0$. 
Combining the two cases proves \eqref{eq:growth}.

Now, our goal is to apply \cref{lem:remez} with $a=1/2$
and $I=[0,1/2]$. For that we need a measurable subset $E\subseteq I$
which has large measure. To that end, we observe that since
$A\subseteq[-1,1]$, either $A_+:=A\cap[0,1]$ or $A\cap[-1,0]$
has measure at least $|A|/2$. Without loss of generality, assume it
is the former. Otherwise, replace $H(z)$ by $H(-z)$ and $A$ by $-A$.
This preserves the assumptions, since the reflected frequencies belong
to $-\Lambda\subseteq[-\tau,\tau]$.

Now, $A_+$ lies in $[0,1]$, but need not lie in the desired interval
$I=[0,1/2]$. We therefore take its preimage under the map
$q(x)=1-4x^2$, which sends $[0,1/2]$ bijectively onto $[0,1]$. It remains to ensure
that this preimage has large measure. Since $|q'(x)|=8x\le 4$ for every $x\in[0,1/2]$, we have that $q$ is $4$-Lipschitz on $[0,1/2]$, and thus $|q(E)|\le 4|E|$ for every measurable $E\subseteq[0,1/2]$. Applying this to $E:=\{x\in[0,1/2]:q(x)\in A_+\}$,
and using $q(E)=A_+$, we obtain
\begin{equation}
\label{eq:EatleastA}
    |A_+|=|q(E)|\le4|E|
    \quad\implies\quad
    |E|\ge\frac{|A_+|}{4}\ge\frac{|A|}{8}.
\end{equation}
Thus the set $E\subseteq I$ on which we will apply \cref{lem:remez} remains quantitatively large.

Now, define
$    f(z):=e^{-4\tau}H(q(z)).$ 
For \(|z|<1\), observe that 
$
    |\operatorname{Im}q(z)|
    =
    4|\operatorname{Im}(z^2)|
    \le4|z|^2
    <4.
$
Hence Eq.~\eqref{eq:growth} gives 
$$
    |f(z)|
    \le
    e^{-4\tau}e^{4\tau}
    =1,
$$
for $|z|\leq 1$, so it is bounded  and we can apply \cref{lem:remez}. 
Also, $
    q(-1/2)=q(1/2)=0,
$~hence
$$
    |f(-1/2)|
    =
    |f(1/2)|
    =
    e^{-4\tau}|H(0)|
    =
    e^{-4\tau}.
$$
Since $q([0,1/2])=[0,1]$, and $|H|\le1$ on $[0,1]$, and
$|H(0)|=1$, we have
$$
    \max_{x\in[0,1/2]}|f(x)|
    =
    e^{-4\tau}.
$$
Moreover, by the definition of \(E\),
$$
    \sup_{x\in E}|f(x)|
    \leq 
    e^{-4\tau}
    \sup_{t\in A_+}|H(t)|.
$$
Applying \cref{lem:remez} to $f$ with
$
    a=1/2$, $
    I=[0,1/2]
$, $
    b=e^{-4\tau}
$ and $\sigma=24\tau$, we get
$$
    e^{-4\tau}
    \le
    \left(\frac4{|E|}\right)^{24\tau}
    e^{-4\tau}
    \sup_{t\in A_+}|H(t)|   \le 
    \left(\frac{32}{|A|}\right)^{24\tau}e^{-4\tau}
    \sup_{t\in A}|H(t)|,
$$
where the second inequality used  Eq.~\eqref{eq:EatleastA}.
Plugging in $H=h/M$, the above inequality implies 
$$
    M
    \le
    \left(\frac{32}{|A|}\right)^{24\tau}
    \sup_{t\in A}|h(t)|,
$$
which proves the theorem. % Eq.~\eqref{eq:conclusion}.
\end{proof}

%% file: learning.tex
\section{Learning Random Circuits with State Access}\label{sec:learning}

\paragraph{Setup.} For simplicity, we will ignore rounding issues and assume that the parameters are integers when required without explicitly writing it throughout this paper.  Consider the $k$-dimensional lattice $\cL = (\Z/w\Z)^k$ where $w = n^{1/k}$.  The lattice has $n$ sites where each site has an $\ell$-qubit local state and gates of locality $2\ell$ act on two neighboring sites. The \emph{$k$-dimensional brickwork architecture} with gates of locality $2\ell$ applies a depth-$d$ circuit as follows: the $d$ layers of gates are indexed by a tuple $(a,s)$ where $a \in [k]$ and $s \in [d/k]$ and are arranged as 
$$(k,d/k),\ldots,(1,d/k),\ldots,(k,2), \ldots, (1,2),(k,1),\ldots,(1,1).$$

These layers are applied to the initial state in order from right to left with the first layer of gates being applied indexed by $(1,1)$. The layer $(a,s)$ consists of $2\ell$-qubit quantum gates on all disjoint pairs of neighboring sites that are matched across direction $a$ as follows: for each value $y$ taken by $x_{-a}$, when $s$ is odd, a gate is applied on sites with locations $(x_a=2i-1, x_{-a}=y)$ and $(x_a=2i, x_{-a}=y)$, while if $s$ is even a gate is applied on sites with locations $(x_a=2i, x_{-a}=y)$ and $(x_a=2i+1, x_{-a}=y)$, where the operations are modulo $w$.

Note that $s$ tracks the number of times the gates are applied along a given direction and the gates alternate between two overlapping perfect matchings depending on the parity of $s$. For example, for any $a$ when $s=1$, the sites where the $a^{\text{th}}$ coordinates are $(1,2),(3,4),\ldots$ are acted upon by the quantum gates, while when $s=2$, the gates act on sites where the $a^{\text{th}}$ coordinates are $(w,1),(2,3),\ldots$. 
\begin{figure}[H]
\centering
\begin{tikzpicture}[
    font=\sffamily\scriptsize,
    text=lcInk,
    line cap=round,
    line join=round,
    lc wire/.style={
        draw=lcInk,
        line width=0.4pt
    },
    lc gate/.style={
        draw=lcGate,
        fill=white,
        line width=0.4pt,
        rounded corners=1.5pt
    },
    lc grid/.style={
        draw=lcMuted,
        line width=0.35pt,
        opacity=0.35
    },
    lc bond/.style={
        draw=lcInk,
        line width=1.6pt
    },
    lc periodic/.style={
        draw=lcInk,
        line width=1.6pt,
        rounded corners=2pt
    }
]

% ============================================================
% (a) One-dimensional lattice
% ============================================================
\begin{scope}[yshift=-2.26cm,x=0.62cm,y=0.93cm]

    \foreach \x in {0,...,7} {
        \draw[lc wire] (\x,0.389) -- (\x,6.611);
        \fill[lcInk] (\x,0.389) circle (2.4pt);
    }

    \foreach \layer in {1,3,5} {
        \foreach \a in {0,2,4,6} {
            \draw[lc gate]
                ({\a-0.393},{\layer-0.297})
                rectangle
                ({\a+1.393},{\layer+0.297});
        }
    }

    \foreach \layer in {2,4,6} {
        \foreach \a in {1,3,5} {
            \draw[lc gate]
                ({\a-0.393},{\layer-0.297})
                rectangle
                ({\a+1.393},{\layer+0.297});
        }
    }

    \node at (3.5,-0.20) {(a) One-dimensional lattice};

\end{scope}

% ============================================================
% (b) Two-dimensional lattice: Layer 1
% ============================================================
\begin{scope}[shift={(6.25,1.90)},x=0.82cm,y=0.82cm]

    \foreach \u in {0,...,3} {
        \draw[lc grid] (\u,0) -- (\u,3);
        \draw[lc grid] (0,\u) -- (3,\u);
    }

    \foreach \y in {0,...,3} {
        \draw[lc bond] (0,\y) -- (1,\y);
        \draw[lc bond] (2,\y) -- (3,\y);
    }

    \foreach \x in {0,...,3} {
        \foreach \y in {0,...,3} {
            \fill[lcInk] (\x,\y) circle (2pt);
        }
    }

    \node at (1.5,-0.65) {(b) Layer 1};

\end{scope}

% ============================================================
% (c) Two-dimensional lattice: Layer 2
% ============================================================
\begin{scope}[shift={(10.55,1.90)},x=0.82cm,y=0.82cm]

    \foreach \u in {0,...,3} {
        \draw[lc grid] (\u,0) -- (\u,3);
        \draw[lc grid] (0,\u) -- (3,\u);
    }

    % Interior horizontal bonds.
    \foreach \y in {0,...,3} {
        \draw[lc bond] (1,\y) -- (2,\y);
    }

    % Periodic horizontal bonds: x=3 is adjacent to x=0.
    \foreach \y in {0,...,3} {
        \draw[lc periodic]
            (3,\y)
            .. controls (3.45,{\y+0.28}) and (-0.45,{\y+0.28}) ..
            (0,\y);
    }

    \foreach \x in {0,...,3} {
        \foreach \y in {0,...,3} {
            \fill[lcInk] (\x,\y) circle (2pt);
        }
    }

    \node at (1.5,-0.65) {(c) Layer 2};

\end{scope}

% ============================================================
% (d) Two-dimensional lattice: Layer 3
% ============================================================
\begin{scope}[shift={(6.25,-1.90)},x=0.82cm,y=0.82cm]

    \foreach \u in {0,...,3} {
        \draw[lc grid] (\u,0) -- (\u,3);
        \draw[lc grid] (0,\u) -- (3,\u);
    }

    \foreach \x in {0,...,3} {
        \draw[lc bond] (\x,0) -- (\x,1);
        \draw[lc bond] (\x,2) -- (\x,3);
    }

    \foreach \x in {0,...,3} {
        \foreach \y in {0,...,3} {
            \fill[lcInk] (\x,\y) circle (2pt);
        }
    }

    \node at (1.5,-0.65) {(d) Layer 3};

\end{scope}

% ============================================================
% (e) Two-dimensional lattice: Layer 4
% ============================================================
\begin{scope}[shift={(10.55,-1.90)},x=0.82cm,y=0.82cm]

    \foreach \u in {0,...,3} {
        \draw[lc grid] (\u,0) -- (\u,3);
        \draw[lc grid] (0,\u) -- (3,\u);
    }

    % Interior vertical bonds.
    \foreach \x in {0,...,3} {
        \draw[lc bond] (\x,1) -- (\x,2);
    }

    % Periodic vertical bonds: y=3 is adjacent to y=0.
    \foreach \x in {0,...,3} {
        \draw[lc periodic]
            (\x,3)
            .. controls ({\x+0.28},3.45) and ({\x+0.28},-0.45) ..
            (\x,0);
    }

    \foreach \x in {0,...,3} {
        \foreach \y in {0,...,3} {
            \fill[lcInk] (\x,\y) circle (2pt);
        }
    }

    \node at (1.5,-0.65) {(e) Layer 4};

\end{scope}

\end{tikzpicture}

\caption{
Brickwork architecture.
Left: a six-layer circuit on a one-dimensional lattice, where each rounded rectangle denotes a two-site gate.
Right: four successive layers on a two-dimensional square lattice with $n=16$. Thick black segments indicate the active two-site gates, and across the four layers every nearest-neighbor edge is used.%\mnote{make these two sides almost the same size or change layout}\qnote{Fixed!}
}
\label{fig:brickwall-architecture}
\end{figure}
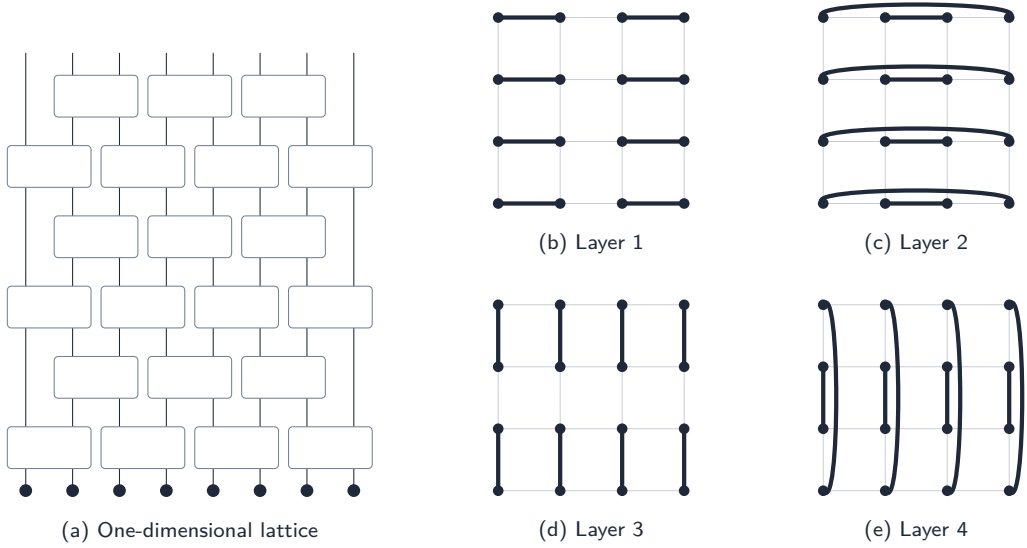

\paragraph{The Learning Model.} We sample a set of $M=64n^2d^2/\delta^2$ unitaries $G_1,\ldots, G_M$ that act on $2\ell$ qubits from the Haar measure on $\U(4^\ell)$. A finite description $\langle G_i \rangle$ of each of these matrices is stored by truncating the real and imaginary parts of each entry to $p=\mathrm{poly}(n)$ bits of precision. These truncated matrices may not be unitary, so we assume a decoding procedure $\textsc{Decode}$ that rounds a matrix description $\langle G_i \rangle$ to its nearest unitary matrix $U_{\langle G_i\rangle}$ by using its polar decomposition: if the (non-unitary) matrix described by a description $\langle G \rangle$ is $\tilde{G}$ then it is rounded to the unitary $\tilde{G}(\tilde{G}^\dag \tilde{G})^{-1/2}$ appearing in the polar decomposition of $\tilde{G}$. The rounded matrix may not be the same as the original matrix before the truncation, but it will be very close in operator norm\footnote{In particular, it can be shown that $\|U_{\langle G\rangle}-G\|_{\mathrm{op}}\le 2\|\tilde G-G\|_{\mathrm{op}}\lesssim 4^\ell 2^{-p}$. If $\|\tilde G-G\|_{\mathrm{op}}<1$, $\tilde G$ is invertible.}.  The rounded unitary matrices $\textsc{Decode}(\langle G_i \rangle)$ form the gateset $\cN$. The gates of the circuit are going to be only chosen from the gateset $\cN$. 

\medskip 
Using the gate set $\cN$ constructed above, we generate a random
$k$-dimensional brickwork circuit $U$ by assigning each gate location
a label sampled independently and uniformly from $[M]$. This label
assignment is independent of the randomness used to construct $\cN$.
At each location assigned label $i$, we place the corresponding
decoded gate $\textsc{Decode}(\langle G_i\rangle)\in\cN$.
The learning algorithm is given copy access to the state
$\ket{\psi_0}:=U\ket{0^\ell}^{\otimes n}$.
The classical description $\langle\cN\rangle$ and the decoding
procedure $\textsc{Decode}$ are public and known to the algorithm,
whereas the sampled label assignment defining $U$ remains unknown.
The learning task is to output the unknown circuit $U$.

We remark on two points regarding the above modeling. First, after averaging over the random choice of $\cN$, the ideal circuit distribution is close in total variation distance\footnote{Conditioned on no gate label being used twice, the ideal gates are independent Haar samples after averaging over the gateset $\cN$, so the total variation distance is at most the probability of a repeated label, which is $O(n^2d^2/M)=O(\delta^2)$ by a union bound over pairs of gate locations.} to the ensemble with independent local Haar random gates. Second, for constant locality, common choices for such problems include fixed universal gate sets~\cite{HL09,OSH20} and $\varepsilon$-nets~\cite{fefferman2024anticoncentrationunitaryhaarmeasure} (see also~\cite{KW} for learning with finite gate sets). At inverse polynomial precision, such nets have polynomial size for constant locality but size $\exp(n^{\Omega(1)})$ for $\Theta(\log n)$ locality. Compiling larger gates into a standard universal set of one- and two-qubit gates instead changes the depth and architecture. We therefore use a public gate set of polynomial size, retaining larger gates as individual operations without requiring coverage of all local unitaries. We note that making the gateset public does not make the problem substantially easier as the search space still consists of $2^{\poly(n)}$ candidate circuits. We also note that for the constant locality case, our algorithm can be adapted to work with a fixed $\varepsilon$-net similar to the setup of \cite{fefferman2024anticoncentrationunitaryhaarmeasure}. This requires some non-trivial work, but the key ideas are the same, so we do not discuss this here to keep the presentation focused.

\medskip
% \learningupperbound*
\learning*

While the guarantee of recovering the circuit $U$ exactly and with only state access may appear surprising, the proof in fact shows that with high probability two different random circuits cannot generate states that are arbitrarily close. A similar guarantee appears in the algorithm of \cite{fefferman2024anticoncentrationunitaryhaarmeasure}.

\paragraph{Learning Algorithm.} 

Let $U = U_d U_{d-1}\ldots U_1$ be the circuit where $U_t$ denote the circuit of gates in the layer $t$ indexed by the tuple $(a,s)$ satisfying $t=(s-1)k+a$. The learning algorithm will learn the layers of the circuit in backward fashion, first learning layer $U_d$, then $U_{d-1}$ and so on. All the layers will be learned exactly with high probability over the choice of the circuit. Thus, we can assume that for any $t \in [d]$, the algorithm can prepare copies of the state $\ket{\psi_t} := U_tU_{t-1}\ldots U_1\ket{0^\ell}^{\otimes n}$ by applying the inverse circuit $U_{t+1}^\dag\ldots U^\dag_d$ to the copies of the input state $\ket{\psi_0} = U\ket{0^\ell}^{\otimes n}$. 

%The top $d-k+1$ layers and the bottom $k$ layers will be learned in different phases, so we will treat them differently. 

Let us assume that the algorithm is learning the layer $U_t$ with $t=(s-1)k+a$ and let $\psi_t = \ketbra{\psi_t}{\psi_t}$ be the state obtained by undoing the layers above. We will omit the subscript $t$ in this description below to simplify notation as the subscript will be clear from the context. Consider a gate $G$ of the circuit which acts on two neighboring sites $B $ and $C$ where $B$ and $C$'s location are given by $(x_a=i,x_{-a}=y)$ and $(x_a=i+1,x_{-a}=y)$ respectively. Consider the site $D$ with location $(x_a=i+2s-1,x_{-a}=y)$. Our starting point is the observation that the backward lightcones of the sites $B$ and $D$ just intersect. In particular, the backward lightcones are disjoint if the gate $G$ was removed from the circuit. This is easy to observe in one dimensional lattice (see \cref{fig:lightcone-removal}) and continues to hold for the higher dimensional brickwork circuits considered here. 

\begin{lemma}[Disjointness of Backward Lightcones] \label{lem:disjointness}
Let $G$ be a gate in a layer $t$ that acts on the two sites $B$ and $D$ described above. Then, if $G$ is removed from the circuit, the backward lightcones of the sites $B$ and $D$ are disjoint. %\mnote{Add picture in the proof overview section in intro}
\end{lemma}
\vspace*{-15pt}
\begin{proof}
    For layers $(a,s)$ where $s=1$, the statement is trivial. When $s \ge 2$, note that only layers in the direction $a$ can change the $a^{\text{th}}$ coordinate and there are exactly $s$ such layers. After removing the gate $G$ from the circuit, the backward lightcone of $B$ can only reach sites $z$ with $z_a \le x_a + s - 2$ while the backward lightcone of $D$ can only reach sites $z$ with $z_a \ge x_a + s-1$. This implies that the lightcones are disjoint in the $a^{\text{th}}$ coordinate, which also proves that they are disjoint. 
\end{proof}

The above implies that the state $\rhodef = G^\dag\psi G$ obtained by removing the gate $G$ from the circuit becomes separable when reduced to the sites $B$ and $D$. That is, $\rhodef_{BD} = \rhodef_{B}\otimes \rhodef_D$. We define the following correlation function for an arbitrary quantum state $\rho$ to measure the distance from separability: 
\begin{align}\label{eqn:corr}
        \corr_{B:D}(\rho) = \|\rho_{BD}-\rho_{B}\otimes \rho_D\|_F^2,
\end{align}
where the norm denotes the Frobenius norm. This quantity can be considered a measure of mutual information between the sites $B$ and $D$ in the state $\rho$. Note that  $\corr_{B:D}(\rhodef) =  0$ because of the disjointness of the backward lightcones.

Our key technical statement is that with high probability, for the original state $\psi$ where the gate $G$ is present in the circuit, this correlation is large.%\mnote{the notation below should say $\psi$ } \qnote{Fixed!}

%\qnote{I made a slight change in the quantifier of the lemma, to make it more align with the main theorem (main theorem is not changed)}
\begin{lemma}[Correlation Gap]
\label{lem:correlation-separation}
With probability at least $1-\delta$ over the choice of the initial
random gateset $\cN$ described above, the following holds: With
probability at least $1-\delta/2$ over the random label assignment
defining $U$, consider any depth $1\leq t\leq d$ and any gate position acting on some qudits $(B, C)$ at depth $t$. Suppose label $p\in [M]$ is assigned to $G$ and let $D$ be the site whose backward lightcone at depth $t$ just intersects the lightcone of the site $B$. For any candidate label $q \in [M]$ and associated gate $G_{q}\in\mathcal N$ acting on sites $(B, C)$, let
\[
\psi^{(G_{q})}:=G_{q}^\dagger \psi G_{q},
\]
where $\psi$ is the state obtained after applying the circuit to the
depth $t$. There exists a deterministic parameter $\Gamma = 2^{-O_k(\ell d)}\operatorname{poly}\!\left(
\frac{\delta}{ndM}
\right)$ such that,

If $p=q$, then:
\[
\operatorname{Corr}_{B:D}(\psi^{(G_{q})})=0.
\]
If $p\neq q$, then:
\[
\operatorname{Corr}_{B:D}(\psi^{(G_{q})})
\geq
\Gamma
\]
%\mnote{TBD: make notation consistent and shorten the statement of the lemma.}\qnote{Now the notation aligns with the previous part}
\end{lemma}

We prove this result using our Carbery-Wright anticoncentration inequality over the unitary group. In particular, the correlation function above can be written as a constant-degree polynomial in the gates of the circuit and our goal ultimately boils down to showing that this polynomial is not too small with some reasonable probability. Proving this requires a careful inductive argument tracking the small-ball exponents carefully. For a polynomial-time algorithm as desired, it is crucial that the small-ball exponent only depends on the degree of the polynomial and not the dimension of the gates which can be very large. Thus, the dimension-free dependence in the small-ball exponent is absolutely essential here to handle gates of larger than constant locality.

The above suggests a natural test to learn the gate $G$ exactly: apply the inverse of each gate in the public gateset to the state $\psi$; if one chooses the correct gate, the correlation becomes zero while for an incorrect guess, the correlation remains at least $2^{-O(\ell d)}$. This can be detected by performing tomography on the reduced system $BD$ with sufficient accuracy. In particular, since the dimension of the system $BC$ is $2^{O(\ell)}$ and the precision required is roughly $2^{-O_k(\ell d)}$, this implies that one such gate can be learned in time $\mathrm{poly}(n,\delta^{-1})2^{O_k(\ell d)}$. 

Combining all the above steps, the following gives a complete description of the algorithm. 

\begin{center}
\begin{algorithmbox}
\label{alg:Learning protocol}
\begin{enumerate}
  \item \textbf{Choose the correct qubits to measure:} Starting from the last layer, consider each gate $BC$ in the current final layer $(a,s)$. Choose site $D$ such that the location of $B$ and $D$ differs in the $a^{\text{th}}$ coordinate by $2s-1$.

  \item \textbf{Find inverting gate with minimum correlation:} For every gate $G$ in the public gateset $\cN$, apply $G^\dag$ to $BC$ and estimate
  $\mathrm{Corr}_{B:D}(G)$ to precision $2^{-O_k(\ell d)}$ with state tomography on the reduced state $\psi_{BD}$. Choose the gate $G$ with the minimum estimated value of $\mathrm{Corr}_{B:D}(G)$.

  \item \textbf{Prepare copies of state prepared by prefix circuit:} After recovering every gate in the current final layer, apply the inverse of all recovered layers $U_t^\dagger U_{t+1}^\dagger\cdots U_d^\dagger$ to fresh copies of the input state $\ket{\psi_0}$ to obtain copies of the state prepared by the prior layers of the circuit.

  \item \textbf{Learn each layer:} Repeat until all layers have been recovered and output the resulting circuit.
\end{enumerate}
\end{algorithmbox}
\end{center}

\bigskip
\begin{proof}[Proof of \cref{thm:learning}] 
Assume that $0<\delta<1$. We first show that a correlation gap
$\Gamma>0$ allows us to recover one gate exactly. For any density
matrices $X,Y$ on $BD$, the reverse triangle inequality gives
\begin{align}
\label{eq:stable-inequality}
    \left|
        \sqrt{\corr_{B:D}(X)}-
        \sqrt{\corr_{B:D}(Y)}
    \right|
    \nonumber
    &\leq
    \left\|X-Y-(X_B\otimes X_D-Y_B\otimes Y_D)\right\|_F\\
    \nonumber
    &\leq
    \|X-Y\|_1+\|(X_B-Y_B)\otimes X_D\|_1+\|Y_B\otimes(X_D-Y_D)\|_1\\
    \nonumber
     &\leq
    \|X-Y\|_1+\|(X_B-Y_B)\|_1\|X_D\|_1+\|Y_B\|_1\|(X_D-Y_D)\|_1\\
    &\leq 3\|X-Y\|_1.
\end{align}
Here the second inequality follows by triangle inequality since
\[
    X_B\otimes X_D-Y_B\otimes Y_D
    =(X_B-Y_B)\otimes X_D
      +Y_B\otimes(X_D-Y_D),
\]
and the last
inequality uses contractivity of the trace norm under partial trace.

Fix a gate location in the current final layer, and let $p$ be its
true label. Suppose the following condition holds:
\begin{align}
\label{cond:correlation-separation}
    \text{For any gate location:}\qquad\mathrm{Corr}_{B:D}(G_p)=0,
    \qquad
    \mathrm{Corr}_{B:D}(G_q)\geq\Gamma
    \quad\text{for every }q\neq p.
\end{align}
For each candidate, perform tomography on the processed reduced
state $\psi_{BD}^{(G_q)}$, returning a density matrix within trace-norm
error $\varepsilon:=\frac{\sqrt{\Gamma}}{12}$
and denote its correlation by
$\widehat{\mathrm{Corr}}_{B:D}(G_q)$. Whenever all these estimates
are accurate, inequality~\eqref{eq:stable-inequality} implies
\[
    \widehat{\mathrm{Corr}}_{B:D}(G_p)
    \leq 9\varepsilon^2=\frac{\Gamma}{16},
    \qquad
    \widehat{\mathrm{Corr}}_{B:D}(G_q)
    \geq (\sqrt{\Gamma}-3\varepsilon)^2
    =\frac{9\Gamma}{16}
    \quad(q\neq p).
\]
Thus choosing the gate with minimum estimated value of $\widehat{\mathrm{Corr}}_{B:D}(G_q)$ over $q\in[M]$ could identify the true label. Moreover, by \cref{lem:correlation-separation}, with probability
at least $1-\delta$ over the choice of $\cN$, the desired condition holds
with probability at least $1-\delta/2$ over the random label assignment.

Hence, the algorithm performs each required state tomography to precision $\frac{\sqrt{\Gamma}}{12}$ with failure probability at most
$\frac{\delta}{ndM}$. By the argument above, if condition~\eqref{cond:correlation-separation} holds and all of the $ndM/2$ state tomography calls succeed, the algorithm would learn the whole circuit correctly.

\noindent\textit{Success Probability:} Fix a gateset $\cN$ satisfying the conclusion of
\cref{lem:correlation-separation}. For this gateset,
condition~\eqref{cond:correlation-separation} fails with probability
at most $\delta/2$ over the random label assignment.
Each tomography call uses fresh copies and has failure probability
at most $\delta/(ndM)$ conditionally on preceding outcomes. A union bound over the at most $ndM/2$ tomography calls and the
failure of condition~\eqref{cond:correlation-separation} gives
\[
    \Pr_{U,\mathrm{meas}}\!\left[
        \text{the algorithm fails}\mid\cN
    \right]
    \leq
    \frac{\delta}{2}
    +\frac{ndM}{2}\cdot\frac{\delta}{ndM}
    =\delta.
\]
Such gatesets occur with probability at least $1-\delta$ over
the choice of $\cN$, proving the required success guarantee.

\noindent\textit{Running time:} Each tomography call acts on a system of dimension
$4^\ell$ and has sample complexity and running time polynomial in
$4^\ell$, $(\frac{\sqrt{\Gamma}}{12})^{-1}$, and $\log(ndM/\delta)$.
Since $M=64n^2d^2/\delta^2$, our choice of $\Gamma$ gives
\[
    \left(\frac{\sqrt{\Gamma}}{12}\right)^{-1}
    =\operatorname{poly}_k(n,d,\delta^{-1})2^{O_k(d\ell)}.
\]
Accounting for at most $ndM$ tomography calls and at most $nd$
known inverse gates applied to each processed copy, the total
number of copies and running time are
\[
    \operatorname{poly}_k(n,d,\delta^{-1})2^{O_k(d\ell)},
\]
as claimed.
\end{proof}

\subsection[Proof of Correlation Gap]
{Proof of Correlation Gap (\cref{lem:correlation-separation})}

Recall that the circuit has two independent sources of randomness:
the initial random gate set $\cN$ and the labels sampled independently
and uniformly from $[M]$ at each gate location. Define $\cE$ to be
the event that all gate locations get distinct labels.
There are $nd/2$ gate locations, and any two distinct locations
receive the same label with probability $1/M$. By union bound, we have
\[
\Pr[\cE^c]
\le
\binom{nd/2}{2}\frac{1}{M}
\le
\frac{n^2d^2}{8M}
=
\frac{\delta^2}{512}.
\]

Moreover, conditioned on the event $\cE$, the random circuit is equivalent to sampling each gate independently and uniformly from the Haar measure on $\U(4^\ell)$, allowing us to apply the anticoncentration
inequality proved in the previous section.

Next, conditioned on the label assignment satisfying property $\cE$, we prove that the theorem holds for every gate location and then use a union bound to conclude. In particular, it suffices to prove that the requirement holds for each gate location at depth $t=d$. The purpose of the following sections is to prove that every incorrect candidate in the gate set induces a
non-negligible two-point correlation with high probability over the random gateset and the associated random circuit. 

% Consider the site $BC$ (which we assumed already is in layer $d$) where the gate acts in the statement of the \cref{lem:correlation-separation} and let $D$ be the site whose backward lightcone just intersects with that of $B$. Since the backward lightcones of the sites $B$ and $D$ just intersect, there are two paths of gates, $L_1$ and $L_2$, that start from the same gate in the first layer (see \cref{fig: two-ray}). We refer to the gates along these paths as the boundary gates. 

Consider the gate on sites $BC$ in layer $t=d$ from \cref{lem:correlation-separation}, and let $D$ be the site whose backward lightcone just intersects that of $B$. Since the backward lightcones of $B$ and $D$ just intersect, there are two paths of gates, $L_1$ and $L_2$, originating from a common gate in the first layer (see \cref{fig: two-ray}). We call the gates along these paths the boundary gates. There are only $O(d)$ boundary gates in total. We prove the stronger statement that, even after fixing all gates outside the boundary arbitrarily, the correlation, viewed as a function of the boundary gates and the true gate, has supremum bounded below by a constant and satisfies a uniform lower-tail bound.

% We prove the following stronger statement: the supremum and lower tail bounds hold uniformly over every
% fixed assignment of the unitary gates outside the boundary.

The proof has three components.
\begin{enumerate}
    \item
   We first construct an explicit assignment of the boundary
    gates that yields a uniform lower bound
    on the supremum of the correlation polynomial. This bound
    holds for every fixed assignment of all remaining gates.

    \item
     We next prove a multigate extension of
    \cref{thm:unitary-smallball} for polynomials of bounded
    degree in each independent Haar-random unitary input. Since the correlation polynomial has constant degree on each boundary gate,
    combining this inequality with the uniform supremum
    lower bound gives a high-probability lower bound on
    the correlation for each incorrect candidate label.
    \item
   Finally, we take a union bound over all incorrect candidate
    labels to establish the desired property at any gate
    location. By a union bound over all gate locations, we complete the proof.
\end{enumerate}

\subsubsection{Lower Bound on \texorpdfstring{$L^\infty$}{L-infinity} Norm of Two-point Correlation Polynomial}
\label{subsubsec:expectation-two-point}

%\mnote{revising this}

Fix a label assignment satisfying $\cE$, a tested gate on $BC$ with
true label $p$, and an incorrect candidate label $q\neq p$. We show that,
for every fixed assignment of the gates outside the boundary, the remaining
gates can be assigned so that the candidate $G_q$ produces a constant
two-point correlation.

\begin{lemma}
\label{lem:bell-routing-witness}
Fix a label assignment satisfying $\cE$, and let $p$ be the true label
of the tested gate on $BC$. Fix an incorrect candidate label $q\neq p$
and an arbitrary assignment of all gateset entries except the true entry
$G_p$ and the entries used at the boundary locations. There is an
assignment of these remaining entries such that
\[
    \psi^{(G_q)}_{BD}
    =
    \ketbra{\Phi^+}{\Phi^+}_{BD},
\]
and hence
\[
    \operatorname{Corr}_{B:D}\bigl(\psi^{(G_q)}\bigr)
    =
    1-\frac{1}{2^{2\ell}}
    \geq \frac34.
\]
\end{lemma}

\begin{proof}
In the circuit with the tested gate on $BC$ deleted, choose the gate at
the common starting point of the two boundary paths $L_1$ and $L_2$ to be a unitary that maps $\ket{0^\ell}\otimes\ket{0^\ell}$ to the maximally entangled state
\[
    \ket{\Phi^+}
    =
    \frac{1}{\sqrt{2^\ell}}
    \sum_{x=0}^{2^\ell-1}\ket{x,x}.
\]
Assign SWAP gates along $L_1$ and $L_2$ to route the two halves of this
state to $C$ and $D$ respectively (see \cref{fig: two-ray}), and assign the identity whenever a
half is meant to remain at its current site. By the definition of the
boundary, all gates that can affect either routed half are included among
the entries being assigned. Thus, the construction is independent of the
fixed gates outside the boundary.

Because the labels are distinct, $p\neq q$ and the tested occurrence of $G_p$ on $BC$ has been deleted, the gate $G_p$ appears nowhere else in the circuit and remains free to choose. Set $G_p=G_q\operatorname{SWAP}_{B,C}$, so that after applying the candidate inverse $G_q^\dagger$, the residual gate at the tested location is
\[
    G_q^\dagger G_p=\operatorname{SWAP}_{B,C}.
\]
This SWAP moves the half of the entangled pair at $C$ to $B$, while the
other half remains at $D$. Consequently,
\[
    \psi^{(G_q)}_{BD}
    =
    \ketbra{\Phi^+}{\Phi^+}_{BD}.
\]
The two one-site marginals are both maximally mixed, so
\[
    \operatorname{Corr}_{B:D}\bigl(\psi^{(G_q)}\bigr)
    =
    \left\|
      \ketbra{\Phi^+}{\Phi^+}
      -\frac{I}{2^{2\ell}}
    \right\|_F^2
    =
    1-\frac{1}{2^{2\ell}}
    \geq\frac34.\qedhere
\]
\end{proof}

\begin{figure}[H]
\centering
\begin{tikzpicture}[
  x=0.62cm,y=0.93cm,
  font=\scriptsize,
  line cap=round,
  line join=round,
  wire/.style={
    draw=black,
    line width=0.8pt
  },
  gate/.style={
    draw=gateblue,
    fill=white,
    line width=0.45pt,
    rounded corners=3pt
  },
  swap/.style={
    draw=black,
    line width=0.65pt
  }
]

% ============================================================
% Sixteen vertical wires.
% Zero-based wire indices:
% B = 3, C = 4, D = 10.
% ============================================================

\foreach \q in {0,...,15} {
  \draw[wire] (\q,0.39) -- (\q,4.86);
}

% Extend B and C above the enlarged BC block.
\foreach \q in {3,4} {
  \draw[wire] (\q,4.86) -- (\q,5.48);
}

% ============================================================
% Ordinary brickwork gates.
% ============================================================

% First and third layers.
\foreach \layer in {1,3} {
  \foreach \a in {0,2,...,14} {
    \draw[gate]
      ({\a-0.39},{\layer-0.30})
      rectangle
      ({\a+1.39},{\layer+0.30});
  }
}

% Second layer.
\foreach \a in {-1,1,...,15} {
  \draw[gate]
    ({\a-0.39},1.70)
    rectangle
    ({\a+1.39},2.30);
}

% Fourth layer, excluding the BC gate.
\foreach \a in {-1,1,5,7,9,11,13,15} {
  \draw[gate]
    ({\a-0.39},3.70)
    rectangle
    ({\a+1.39},4.30);
}

% ============================================================
% Color the gates on L_1 and L_2.
% ============================================================

% L_1 gates: (5,6) at layer 2 and (4,5) at layer 3.
\foreach \a/\layer in {5/2,4/3} {
  \draw[
    draw=lcBlue,
    fill=lcBlueFill,
    line width=0.8pt,
    rounded corners=3pt
  ]
    ({\a-0.39},{\layer-0.30})
    rectangle
    ({\a+1.39},{\layer+0.30});
}

% L_2 gates: (7,8) at layer 2, (8,9) at layer 3, (9,10) at layer 4.
\foreach \a/\layer in {7/2,8/3,9/4} {
  \draw[
    draw=lcOrange,
    fill=lcOrangeFill,
    line width=0.8pt,
    rounded corners=3pt
  ]
    ({\a-0.39},{\layer-0.30})
    rectangle
    ({\a+1.39},{\layer+0.30});
}

% ============================================================
% U_Bell gate -- PURPLE.
% ============================================================

\draw[
  draw=BellPurple,
  fill=BellPurpleFill,
  line width=0.8pt,
  rounded corners=3pt
]
  (5.61,0.70) rectangle (7.39,1.30);

\node[
  text=BellPurple,
  font=\large
] at (6.5,1) {$U_{\mathrm{Bell}}$};

% ============================================================
% Enlarged BC block.
% ============================================================

\draw[
  gate,
  rounded corners=6pt
]
  (2.47,3.54) rectangle (4.52,5.13);

% ------------------------------------------------------------
% G_q -- GREEN (upper dotted box).
% ------------------------------------------------------------

\draw[
  draw=GqGreen,
  fill=GqGreenFill,
  line width=0.75pt,
  dash pattern=on 0pt off 0.7pt,
  rounded corners=3pt
]
  (2.61,4.43) rectangle (4.41,5.02);

\node[
  text=GqGreen,
  font=\normalsize
] at (3.51,4.72) {$G_q^\dag$};

% ------------------------------------------------------------
% Lower shaded box belongs to L_1 -- BLUE.
% ------------------------------------------------------------

\draw[
  draw=lcBlue,
  fill=lcBlueFill,
  line width=0.75pt,
  dash pattern=on 0pt off 0.7pt,
  rounded corners=3pt
]
  (2.61,3.70) rectangle (4.41,4.30);

\node[
  anchor=north west,
  text=lcBlue,
  font=\normalsize
] at (2.88,4.29) {$G_p$};

% ============================================================
% Large crossed wires inside the enlarged BC block.
% ============================================================

\draw[swap]
  (3.03,4.88)
  -- (3.03,4.55)
  -- (3.91,4.18)
  -- (3.91,3.85);

\draw[swap]
  (3.91,4.88)
  -- (3.91,4.55)
  -- (3.03,4.18)
  -- (3.03,3.85);

% ============================================================
% Small SWAP symbols.
% ============================================================

\foreach \sx/\sy in {
  5.5/2,
  7.5/2,
  4.5/3,
  8.5/3,
  9.5/4
} {

  \draw[swap]
    ({\sx-0.45},{\sy+0.16})
    -- ({\sx-0.45},{\sy+0.055})
    -- ({\sx+0.45},{\sy-0.055})
    -- ({\sx+0.45},{\sy-0.16});

  \draw[swap]
    ({\sx+0.45},{\sy+0.16})
    -- ({\sx+0.45},{\sy+0.055})
    -- ({\sx-0.45},{\sy-0.055})
    -- ({\sx-0.45},{\sy-0.16});
}

% ============================================================
% Input dots.
% ============================================================

\foreach \q in {0,...,15} {
  \fill[black] (\q,0.39) circle (3.2pt);
}

% ============================================================
% Output labels.
% ============================================================

\node at (3,5.68)  {$B$};
\node at (4,5.68)  {$C$};
\node at (10,5.08) {$D$};

% ============================================================
% Legend.
% ============================================================

\begin{scope}[shift={(3.8,6.20)}]

  % L1 gate
  \draw[
    draw=lcBlue,
    fill=lcBlueFill,
    line width=0.7pt,
    rounded corners=1.5pt
  ]
    (0.00,-0.16) rectangle (0.70,0.16);

  \node[anchor=west] at (0.85,0) {$L_1$};

  % L2 gate
  \draw[
    draw=lcOrange,
    fill=lcOrangeFill,
    line width=0.7pt,
    rounded corners=1.5pt
  ]
    (2.10,-0.16) rectangle (2.80,0.16);

  \node[anchor=west] at (2.95,0) {$L_2$};

  % U_Bell gate
  \draw[
    draw=BellPurple,
    fill=BellPurpleFill,
    line width=0.7pt,
    rounded corners=1.5pt
  ]
    (4.20,-0.16) rectangle (4.90,0.16);

 \node[
  anchor=west,
  font=\tiny
] at (5.05,0) {$U_{\mathrm{Bell}}$};

  % G_q gate
  \draw[
    draw=GqGreen,
    fill=GqGreenFill,
    line width=0.7pt,
    rounded corners=1.5pt
  ]
    (7.20,-0.16) rectangle (7.90,0.16);

  \node[anchor=west] at (8.05,0) {$G_q^\dag$};
\end{scope}
\end{tikzpicture}
\caption{Example of a boundary-gate assignment on a one-dimensional lattice. The gate $U_{\mathrm{Bell}}$ prepares a generalized Bell state and the smaller crossings represent SWAP gates used to route its two halves. At the tested location $BC$, the true gate is chosen as $G_p=G_q\operatorname{SWAP}_{B,C}$. Applying the candidate inverse $G_q^\dagger$ therefore leaves the residual SWAP shown by the enlarged crossing.}
% Example of a boundary gate assignment on a one dimensional lattice. The gate ($U_{\mathrm{Bell}}$) creates a generalized Bell state and each pair of crossing lines represents a SWAP gate.} %\mnote{this caption is not gramatically correct.}\qnote{Fixed!}}
\label{fig: two-ray}
\end{figure}
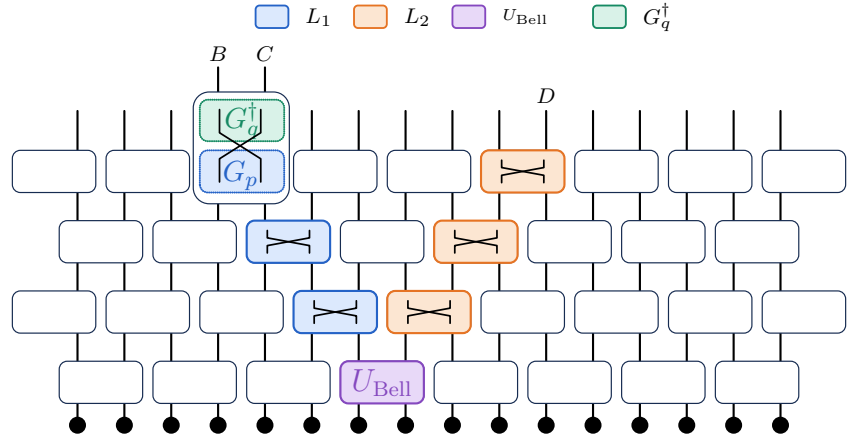

\subsubsection{Anti-concentration for the Two-point Correlation Function}
\label{subsubsec:anticoncentration-two-point}
\cref{thm:unitary-smallball} provides a useful tool that transform the $L_{\infty}$ lower bound into a high-probability argument. To begin with, the following extension works on anticoncentration of a
polynomial that takes several independent Haar random unitaries. 
%\mnote{fix the constants with the updated anticoncentration constants}

\begin{lemma}[Multigate anticoncentration]
\label{thm:multi-unitary-anti-concentration}
Let $P:\U(m)^r\to\C$ be a nonzero complex-valued polynomial. For each $j\in[r]$, suppose $P$ has total degree at most
$d$ in the entries of $U_j$ and $\overline{U_j}$ when all
other unitaries are fixed. Then, for independent $U_1,\ldots,U_r\sim\haar{m}$ and every
$\eps\in(0,1)$,
\[
    \Pr\!\left[
        |P(U_1,\ldots,U_r)|
       \leq \eps \|P\|_\infty
    \right]
    \leq
    (1+64m^2)^r
    \eps^{\frac{1}{48\pi d}}.
\]
\end{lemma}

\begin{proof}
We choose $\widetilde U_1,\ldots,\widetilde U_r$ such that
$|P(\widetilde U_1,\ldots,\widetilde U_r)|=\|P\|_\infty$.
We also write $U^{(j)}:=U_j$ and $F:=|P(U_1,\ldots,U_r)|$.
For simplicity, set $c:=\frac1{24\pi d}, A:=64m^2$. Consider the filtration
\[
    \mathcal F_j
    :=
    \sigma\!\left(U^{(1)},\ldots,U^{(j)}\right).
\]
Define
\[
    H_j
    :=
   \left|P(U^{(1)},\ldots,U^{(j)},\widetilde U_{j+1},\ldots,\widetilde U_r)\right|,
    \qquad j=0,1,\ldots,r.
\]
Then we have
\[
   H_0=\|P\|_\infty,
    \qquad
    H_r=F,
    \qquad
    H_{j-1}
   \leq
   \sup_{U^{(j)}}H_j.
\]

We use the anticoncentration inequality recursively.
Fix $j\in [r]$, and condition on
$U^{(1)},\ldots,U^{(j-1)}$. As a function of $U^{(j)}$, the
random variable $H_j$ is the absolute value of a polynomial of total degree
at most $d$ in the entries of $U^{(j)}$ and $\bar U^{(j)}$. Its supremum over $U^{(j)}$, conditionally on
$\mathcal F_{j-1}$, is at least $H_{j-1}$. Hence, \cref{thm:unitary-smallball} gives
\[
    \Pr\!\left[
        H_j\leq uH_{j-1}
        \,\middle|\,\mathcal F_{j-1}
    \right]
    \leq Au^c
\]
for $u\in(0,1)$. Now define the ratio
\[
    R_j:=\frac{H_j}{H_{j-1}}.
\]
The previous estimate gives
\[
    \Pr\!\left[
        R_j\leq u
        \,\middle|\,\mathcal F_{j-1}
    \right]
    \leq Au^c.
\]

We now convert this lower-tail estimate into a negative-moment
estimate. Let $0<v<c$. Then
\begin{align*}
    \mathbb E\!\left[
        R_j^{-v}\,\middle|\,\mathcal F_{j-1}
    \right]
    &=
    \int_0^\infty
    \Pr\!\left[
        R_j^{-v}\geq s
        \,\middle|\,\mathcal F_{j-1}
    \right]\,ds\\
    &\leq
    1+\int_1^\infty
    \Pr\!\left[
        R_j\leq s^{-1/v}
        \,\middle|\,\mathcal F_{j-1}
    \right]\,ds\\
    &\leq
    1+A\int_1^\infty s^{-c/v}\,ds\\
    &=
    1+\frac{Av}{c-v}.
\end{align*}
Taking $v=c/2=1/(48\pi d)$, we obtain
\[
    \mathbb E\!\left[
        R_j^{-c/2}\,\middle|\,\mathcal F_{j-1}
    \right]
    \leq1+A.
\]

Next, observe that
\[
    \frac{F}{\|P\|_\infty}
    =
    \frac{H_r}{H_0}
    =
    \prod_{j=1}^r\frac{H_j}{H_{j-1}}
    =
    \prod_{j=1}^rR_j.
\]
Therefore, we have:
\begin{align*}
    \mathbb E\!\left[
        \left(\frac{F}{\|P\|_\infty}\right)^{-c/2}
    \right]
    &=
    \mathbb E\!\left[\prod_{j=1}^rR_j^{-c/2}\right]\\
    &\leq(1+A)^r.
\end{align*}
Finally, by Markov's inequality,
\begin{align*}
    \Pr\!\left[F\leq\alpha\|P\|_\infty\right]
    &=
    \Pr\!\left[
        \left(\frac{F}{\|P\|_\infty}\right)^{-c/2}
        \geq\alpha^{-c/2}
    \right]\\
    &\leq
    \alpha^{c/2}
    \mathbb E\!\left[
        \left(\frac{F}{\|P\|_\infty}\right)^{-c/2}
    \right]\\
    &\leq
    (1+A)^r\alpha^{c/2}\\
    &=
    (1+64m^2)^r\alpha^{1/(48\pi d)}.\qedhere
\end{align*}
\end{proof}

We now apply \cref{thm:multi-unitary-anti-concentration} to
the two-point correlation polynomial. Now, fix a label assignment satisfying $\cE$ and an incorrect
candidate label $q\neq p$.
Fix an arbitrary assignment of all gate-set entries
other than boundary gates and $G_q$. Let $V_1,\ldots,V_r$
denote the remaining unitaries unfixed in the gateset, where we have $r\leq 2kd$ and each unitary is sampled independently from continuous haar-random unitaries.

We first bound the degree of $\operatorname{Corr}_{B:D}(\psi^{(G_q)})$ for each unitary. Fix $t\in[r]$, and hold all entries except
$V_t$ fixed. If $V_t$ is not the candidate entry $G_q$, it occurs
exactly once in the processed circuit. Every matrix
entry of $\psi^{(G_q)}_{BD}$ contains at most one factor of $V_t$ and one factor of $V_t^\dagger$. Taking marginals only sums matrix entries, which does not increase the degree. Consequently, every matrix entry has degree at most 2. Then, every matrix entry of $\psi^{(G_q)}_{BD} - \psi^{(G_q)}_B\otimes\psi^{(G_q)}_D$
therefore has degree at most 4, and
\[
    \operatorname{Corr}_{B:D}(\psi^{(G_q)})
    =
    \norm{
        \psi^{(G_q)}_{BD}
        -
        \psi^{(G_q)}_B\otimes\psi^{(G_q)}_D
    }_F^2
\]
has degree at most 8.

If the candidate label $q$ appears on the boundary and
$V_t=G_q$, the same entry occurs once at its circuit location
and once through the candidate inverse, which give the degree at most $16$. Thus, the two-point correlation function is a nonnegative real-valued polynomial taking unitaries input $V_1,\cdots, V_r$ with each
degree at most $16$. Therefore the total degree of each unitary is at most $16$. %\mnote{bidegree -- do you need this terminiology. Is so define it in prelims or before using, other use just degree or total degree?}\qnote{It's fixed now}

Let $\widetilde V_1,\ldots,\widetilde V_r$ be the specific assignment from
\cref{lem:bell-routing-witness}. Also, we write the polynomial
\[
    F(V_1,\ldots,V_r):=\operatorname{Corr}_{B:D}(\psi^{(G_q)}).
\]
The witness gives
\[
   \|F\|_\infty\geq
    F(\widetilde V_1,\ldots,\widetilde V_r)
    =
    1-\frac{1}{2^{2\ell}}
    \geq\frac34.
\]
Above all, we have the following lemma:

\begin{lemma}[High-probability lower bound on two-point correlation]
\label{lem:two-point-high-probability}
Fix a label assignment satisfying $\cE$ and an incorrect
candidate label $q\neq p$. For every $\alpha>0$,
\begin{equation}
    \Pr\!\left[
        \operatorname{Corr}_{B:D}(\psi^{(G_q)})\leq\alpha 
    \right]
    \leq C_2^d\left(\frac{\alpha}{C_0}\right)^{C_1}.
    \label{eq:two-point-correlation-lower-tail}
\end{equation}
The probability is over unitaries in the gateset, $C_0, C_1$ are some absolute constants and $C_2:= (1+64\cdot 2^{4\ell})^{2k}$.
\end{lemma}

\begin{proof}
To begin with, we fix any arbitrary assignment of all gates other than the boundary gates and $G_p$. Denote remaining unitaries $V_1,\ldots,V_r$, which are
independent Haar-random elements of $\U(4^{\ell})$.
By \cref{lem:bell-routing-witness},
\[
   \sup_{V_1,\ldots,V_r}\operatorname{Corr}_{B:D}(\psi^{(G_q)})
   \geq 1-2^{-2\ell} \geq \frac34
\]

The quantity $\operatorname{Corr}_{B:D}(\psi^{(G_q)})$ is nonnegative and has total
degree at most $16$ in each individual free entry and its
complex conjugates. Applying
\cref{thm:multi-unitary-anti-concentration} with
$m=4^{\ell}$ gives
\begin{align*}
    \Pr\!\left[
        \operatorname{Corr}_{B:D}(\psi^{(G_q)})\leq\alpha 
    \right] &\leq \Pr\!\left[
        \operatorname{Corr}_{B:D}(\psi^{(G_q)})
        \leq
        \frac{4\alpha}{3}\,
       \sup_{V_1,\ldots,V_r}\operatorname{Corr}_{B:D}(\psi^{(G_q)})
    \right] \\
    &\leq
        \left(1+64\cdot2^{4\ell}\right)^{2kd}
     \left(\frac{4\alpha}{3}\right)^{1/(768\pi)}
\end{align*}
%\snote{fixed all constants except the one above, so how did 1280 come about, it should have been $1/48\pi d$, what is d here, is it $d=2$?}\qnote{Sorry there is a change of the parameter - it's fixed now}
where the first inequality uses $C_0\leq\sup_{V_1,\ldots,V_r}\operatorname{Corr}_{B:D}(\psi^{(G_q)})$, last inequality uses $r\leq 2kd$. Taking $C_0:=\frac34$ and $C_1:=1/(768\pi)$ gives the desired bound.
\end{proof}

\subsubsection{Putting Everything Together}

\begin{proof}[Proof of \cref{lem:correlation-separation}]
Assume that $0<\delta<1$. Set
\[
    \Gamma
    :=
    C_0
    \left(\frac{\delta^2}{2ndMC_2^d}\right)^{1/C_1}.
\]
Fix any label assignment satisfying $\cE$. For any gate location, suppose its depth is $t\leq d$. Let $p$ be
the true label at the chosen gate location. By
\cref{lem:two-point-high-probability} and a union bound over
the $M-1$ incorrect labels,
\begin{align*}
    \Pr_{\cN}\!\left[
        \exists q\neq p:\ \operatorname{Corr}_{B:D}(\psi^{(G_q)})\leq\Gamma
    \right]
    &\leq
    \sum_{q\neq p}
    \Pr_{\cN}\!\left[\operatorname{Corr}_{B:D}(\psi^{(G_q)})\leq\Gamma\right] &&(\text{By union bound})\\
    &\leq
    \sum_{q\neq p}
    C_2^t\left(\frac{\Gamma}{C_0}\right)^{C_1} &&(\text{By \cref{lem:two-point-high-probability}})\\
    &\leq
    (M-1)C_2^d
    \left[
        \left(\frac{\delta^2}{2ndMC_2^d}\right)^{1/C_1}
    \right]^{C_1}\\
    &=
    (M-1)C_2^d\frac{\delta^2}{2ndMC_2^d}\\
    &=
    \frac{M-1}{M}\cdot\frac{\delta^2}{2nd}
    \leq\frac{\delta^2}{2nd}.
\end{align*}
There are $nd/2$ gate locations in total. Define the following condition:
\begin{align}
\label{cond:required-condition}
    \text{For every gate location,}\quad
        \operatorname{Corr}_{B:D}(\psi^{(G_q)})>\Gamma
        \text{ for every }q\neq p
\end{align}
Therefore, a union
bound over all gate locations, followed by averaging over label
assignments satisfying $\cE$, gives
\[
    \Pr_{\cN,U}\!\left[\text{condition~\eqref{cond:required-condition} holds}
        \,\middle|\,\cE
    \right]
    \geq1- \frac{\delta^2}{2nd}\cdot \frac{nd}{2} =1-\frac{\delta^2}{4}.
\]

By the law of total probability,
\[
    \Pr_{\cN,U}\!\left[
        \text{condition~\eqref{cond:required-condition} fails}
    \right]
    \leq \Pr[\cE^c]+\frac{\delta^2}{4}
    \leq \frac{\delta^2}{512}+\frac{\delta^2}{4}
    \leq \frac{\delta^2}{2}.
\]
Applying Markov's inequality gives
\begin{align*}
    \Pr_{\cN}\!\left[
        \Pr_U\!\left[
            \text{condition~\eqref{cond:required-condition} fails}
            \mid\cN
        \right]>\frac{\delta}{2}
    \right]
    &\leq
    \frac{2}{\delta}\,
    \mathbb{E}_{\cN}\!\left[
        \Pr_U\!\left[
            \text{condition~\eqref{cond:required-condition} fails}
            \mid\cN
        \right]
    \right]\\
    &=
    \frac{2}{\delta}\Pr_{\cN,U}\!\left[
        \text{condition~\eqref{cond:required-condition} fails}
    \right] 
    \leq
    \delta
\end{align*}

For the true candidate, \cref{lem:disjointness} gives
$\operatorname{Corr}_{B:D}(\psi^{(G_p)})=0$ deterministically. Hence, with probability at least $1-\delta$ over $\cN$, the true label
is the unique minimizer with probability at least $1-\delta/2$ over the random label assignment.
Finally,
\[
    \Gamma
    =
    2^{-O_k(\ell d)}
    \text{poly}\left(\frac{\delta^2}{ndM}\right)
\]
which has the required dependence.
\end{proof}